\documentclass[runningheads]{llncs}
\usepackage[utf8]{inputenc}
\usepackage{macros}

\title{Roots in the semiring of finite deterministic dynamical systems}

\author{
François Doré\inst{1}
\and
Kévin Perrot\inst{2}
\and
Antonio E. Porreca\inst{2}
\and\\
Sara Riva\inst{3}
\and
Marius Rolland\inst{2}
}

\institute{
Université Côte d'Azur, CNRS, UMR 7271 I3S, France
\and
Aix-Marseille Université, CNRS, UMR 7020 LIS, Marseille, France
\and
Univ. Lille, CNRS, Centrale Lille, UMR 9189 CRIStAL, Lille, France
}
\authorrunning{F. Doré \and K. Perrot \and A.E. Porreca \and S. Riva \and M. Rolland}

\date{March 2024}

\begin{document}

\maketitle

    \begin{abstract}
Finite discrete-time dynamical systems (FDDS) model phenomena that evolve deterministically in discrete time. It is possible to define sum and product operations on these systems (disjoint union and direct product, respectively) giving a commutative semiring. This algebraic structure led to several works employing polynomial equations to model hypotheses on phenomena modelled using FDDS.
To solve these equations, algorithms for performing the division and computing $k$-th roots are needed.
In this paper, we propose two polynomial algorithms for these tasks, under the condition that the result is a connected FDDS. This ultimately leads to an efficient solution to equations of the type $AX^k=B$ for connected $X$.
These results are some of the important final steps for solving more general polynomial equations on FDDS.
    \end{abstract}
    
    \keywords{discrete dynamical systems, root of graph direct product}
    
	\section{Introduction}
	Finite discrete-time dynamical systems (FDDS) are pairs $(X,f)$ where $X$ is a finite set of states and $f: X \to X$ is a transition function (where no ambiguity arises, we will usually denote $(X,f)$ simply as $X$).
	These systems emerge from the analysis of concrete models such as Boolean networks \cite{gershenson2004random_bn,automata_book}
	and are applied to biology \cite{thomas1990biological_feedback,thomas1973genetic_control_circuits,bernot2013computational_biology} to represent, for example, genetic regulatory networks or epidemic models. 
    We can find them also in chemistry \cite{reaction_systems}, to represent the evolution over discrete time of chemical reactions, or information theory \cite{gadouleau2011graph_entropy}.
	
    We can identify dynamical systems with their transition graph, which have uniform outgoing degree one (these are also known as functional digraphs). Their general shape is a collection of cycles with a finite number of directed trees (with arcs pointing towards the root, \ie in-trees) anchored to them by the root. The nodes inside the cycles are periodic states, while the others are transient states.
	
	The set $(\mathbb{D}, +, \times)$ of FDDS taken up to isomorphism with the disjoint union as a \emph{sum} operation (corresponding to the alternative execution of two systems) and the \emph{direct product} \cite{hammack2011graph_product_book} (corresponding to synchronous execution) is a commutative semiring \cite{dorigatti2018polynomial}. 
	However, this semiring is not factorial, \ie a FFDS admits, in general, multiple factorizations into irreducibles.
	For this reason, the structure of product is more complex compared to other semirings such as the natural numbers, and its understanding remains limited.
    We are still unable to characterize or efficiently detect the FDDS obtained by parallel execution of smaller FDDS.  
    
    Some literature analyzes this problem limited to periodic behaviours, \ie to FDDS with permutations as their transition function \cite{dennunzio2019solving,formenti2021mdds,dennunzio2023dds_journal}. 
    Studying these restricted FDDS is justified by the fact that they correspond to the stable, asymptotic behaviour of the system.
    However, transient behaviour is more vast and various when modelling phenomena such as those from, for example, biology or physics.
    FDDS with a single fixed point have also been investigated \cite{gadouleau2022factorisation_dds} focusing more on the transient behaviours. 
    Nevertheless, we cannot investigate general FDDS through a simple combination of these two techniques. 
    
    A direction for reducing the complexity of the decomposition problem is finding an efficient algorithm for equations of the form $AX = B$, \ie for dividing FDDS.
	The problem is trivially in $\NP$, but we do not know its exact complexity (\emph{e.g.}, $\NP$-hard, $\GI$, or $\P$).
	However, \cite{dennunzio2024resolution_permutation} proved that we can solve these equations in polynomial time if $A$ and $B$ are certain classes of \emph{permutations}, \ie FDDS without transient states. Nevertheless, the complexity of more general cases is unknown even for permutations. 
    
	Another direction is to propose an efficient algorithm for the computation of roots over FDDS. Since \cite{gadouleau2022factorisation_dds}, we are aware of the uniqueness of the solution of $k$-th roots, but once again we do not know the exact complexity of the problem beyond a trivial $\NP$ upper bound.
	
	In this paper, we will exploit the notion of unroll introduced in \cite{gadouleau2022factorisation_dds} to address the division and the root problems in the specific case where $X$ is connected (\ie the graph of $X$ contains just one connected component) and we will also approach the unroll division problem. 
    More precisely, we start by showing that we can compute in polynomial time an FDDS $X$ such that $AX = B$, if any exists \ref{section:division_FDDS}.
    We also show that we can compute in polynomial time, given an FDDS $A$ and a strictly positive integer $k$, a connected FDDS $X$ such that $X^k = A$, if any exists \ref{section:root_FDDS}.
    These two last contributions naturally lead to a solution to the more general equation $AX^k=B$.
    Finally, we prove that the can find in polynomial time an unroll $\forest{x}$ such that $\unroll{A} \forest{x} = \unroll{B}$. Thanks to this result, we will deduce we can find $X$ a FDDS (for two specific form of $X$) such that $AX = B$ \ref{section:divison_unroll}.

		\section{Definitions}
    
    In the following, we will refer to the in-trees constituting the transient behaviour of FDDS just as \emph{trees} for simplicity. An FDDS has a set of weakly connected components, each containing a unique cycle.
    In the following, we will refer to FDDS with only one component as \emph{connected}.
    
	In literature, two operations over FDDS have been considered: the \emph{sum} (the disjoint union of the components of two systems) and the \emph{product} (direct product \cite{hammack2011graph_product_book} of their transition graphs).
    Let us recall that, given two digraphs $A = (V, E)$ and $B= (V', E')$, their product $A\times B$ is a digraph where the set of nodes is $V\times V' $ and the set of edges is $\set{((v, v'), (u, u')) \mid (v, u) \in E, (v', u') \in E'}$.
    When applied to the transition graphs of two connected FDDS with cycle lengths respectively $p$ and $p'$, this operation generates $\gcd(p,p')$ components with cycles of length $\lcm(p,p')$ \cite{dorigatti2018polynomial,hammack2011graph_product_book}.
    
	Let us recall the notion of unroll of dynamical systems introduced in~\cite{gadouleau2022factorisation_dds}. We will denote trees and forests using bold letters (in lower and upper case respectively) to distinguish them from FDDS.
	
	\begin{definition}[Unroll]
		Let $A$ be an FDDS $(X,f)$.
		For each state $u\in X$ and $k \in \mathbb{N}$, we denote by $f^{-k}(u) = \{ v \in X \mid f^k(v) = u \}$ the set of $k$-th preimages of $u$. 
		For each $u$ in a cycle of $A$, we call the \emph{unroll tree of $A$ in $u$} the infinite tree $\tree{t}_u = (V,E)$ having vertices $V = \{(s,k) \mid s \in f^{-k}(u), k \in \mathbb{N}\}$ and edges
		$E = \big\{ \big((v,k),(f(v),k-1) \big) \big\} \subseteq V^2$.
		We call \emph{unroll of $A$}, denoted $\unroll{A}$, the set of its unroll trees.
	\end{definition}

    Unroll trees have exactly one infinite branch on which the trees representing transient behaviour hook and repeat periodically.
    Remark that the forest given by the unroll of a connected FDDS may contain isomorphic trees and this results from symmetries in the original graph. 
    
	This transformation from an FDDS to its unroll has already proved successful in studying operations (particularly the product operation) at the level of transient behaviours. 
    Indeed, the sum (disjoint union) of two unrolls corresponds to the unroll of the sum of the FDDS; formally, $\unroll{A}+\unroll{A'}= \unroll{A+A'}$.
    For the product, it has been shown that it is possible to define an equivalent product over unrolls for which $\unroll{A}\times \unroll{A'}= \unroll{A\times A'}$.
    Here and in the following, the equality sign will denote graph isomorphism.
    
    Let us formally define the product of trees to be applied over the unroll of two FDDS. 
    Since it is known that the product distributes over the different trees of the two unrolls \cite{dore2023decomposition}, it suffices to define the product between two trees. Intuitively, this product is the direct product applied layer by layer.
    To define it, we let $\depth{v}$ be the distance of the node from the root of the tree. 
    
    \begin{definition}[Product of trees]\label{prodintrees} Consider two trees $\tree{t}_1=(V_1,E_1)$ and\linebreak $\tree{t}_2=(V_2,E_2)$ with roots $r_1$ and $r_2$, respectively. Their \emph{product} is the tree $\tree{t}_1 \times \tree{t}_2=(V,E)$ such that 
    $V=\set{(v,u)\in V_1\times V_2 \mid \depth{v}=\depth{u}}$ and
    $E=\set{((v,u),(v',u')) \mid (v,u)\in V, (v,v')\in E_1, (u,u')\in E_2}$.
    \end{definition}

    In the following, we use a total order $\le$ on finite trees introduced in \cite{gadouleau2022factorisation_dds}, which is compatible with the product, that is, if $\tree{t}_1 \le \tree{t}_2$ then $\tree{t}_1 \tree{t} \le \tree{t}_2 \tree{t}$ for all tree $\tree{t}$.
    Let us briefly recall that this ordering is based on a vector obtained from concatenating the incoming degrees of nodes visited through a BFS. During graph traversal, child nodes (preimages in our case), are sorted recursively according to this very order, resulting in a deterministic computation of the vector. 

    We will also need the notion of \emph{depth} for finite trees and forests. The depth of a finite tree is the length of its longest branch. For a forest, it is the maximum depth of its trees.
    In the case of unrolls, which have infinite paths, we can adopt the notion of depth of a dynamical system (that is, the largest depth among the trees rooted in one of its periodic states).
    For an unroll tree $\tree{t}$, its depth is the depth of a connected FDDS $A$ such that $\tree{t} \in \unroll{A}$.
    See Figure~\ref{fig:unroll}.
	\begin{figure}[tb]
    \centering
    \begin{tikzpicture}[xscale=.84,yscale=.9]
    \begin{scope}
        \begin{oodgraph}
        \addcycle[rotation angle=90,xshift=0cm,yshift=1cm,nodes prefix = a]{1};
        \addbeard[attach node = a->1]{1};
        \addbeard[attach node = a->1->1,opening angle=30]{2};
        \addcycle[rotation angle=90,xshift=1cm,yshift=1.5cm,nodes prefix = b]{2};
        \addbeard[attach node = b->1,opening angle=30]{2};
        \addbeard[attach node = b->1->2]{1};
        \addbeard[attach node = b->2]{3};
        \end{oodgraph}
        \node[above right=-.1 of a->1]{$v_0$};
        \node[right=0 of b->1]{$v_1$};
        \node[right=0 of b->2]{$v_2$};
        \node[left=0 of b->1->2->1]{$v_3$};
        \node[scale=2]() at (-1.25cm,2cm){$\mathcal{U}$};
        \draw (-.5,0) to [leftp] (-.5,4);
        \draw (1.75,0) to [rightp] (1.75,4);
        \node[scale=2]at(2.75cm,2cm){$=$};
    \end{scope}
        
    \tikzset{
    cnode/.style={circle,outer sep=0,inner sep=2}
    }
    \def\j{0}
    \def\dcl{-.1}
    
    \begin{scope}[xshift=3.25cm,xscale=.1,yscale=.8]
        \foreach \x/\y [count=\v from 0] in {10.5/0,1.5/1,11.5/1,0.5/2,2.5/2,5.5/2,13.5/2,4.5/3,6.5/3,9.5/3,15.5/3,8.5/4,10.5/4,13.5/4,17.5/4,12.5/5,14.5/5,16.5/5,18.5/5}
        {\node[cnode](\v)at(\x,\y){$\bullet$};}
        \foreach \u/\v in {1/0,2/0,3/1,4/1,5/2,6/2,7/5,8/5,9/6,10/6,11/9,12/9,13/10,14/10,15/13,16/13,17/14,18/14}
        {\draw[ddsedge](\u)--(\v);}
        \foreach \v [count=\i from 0] in {0,2,6,10,14,18}{\node[txtnode,right=\dcl of \v]{$(v_0,\i)$};}
        
        \draw[ddsedge, dashed]($(18)+(-1,1)$)--(18);
        \draw[ddsedge, dashed]($(18)+(1,1)$)--(18);
    \end{scope}

    \begin{scope}[xshift=5.825cm,xscale=.1,yscale=.8]
        \foreach \x/\y [count=\v from 0] in {13.5/0,9.5/1,7.5/1,14.5/1,5.5/2,9.5/2,11.5/2,13.5/2,17.5/2,14.5/3,12.5/3,19.5/3,10.5/4,14.5/4,16.5/4,18.5/4,22.5/4,20.5/5,22.5/5,24.5/5}
        {\node[cnode](\v)at(\x,\y){$\bullet$};}
        \foreach \u/\v in {1/0,2/0,3/0,4/2,5/3,6/3,7/3,8/3,9/8,10/8,11/8,12/10,13/11,14/11,15/11,16/11,17/16,18/16,19/16}
        {\draw[ddsedge](\u)--(\v);}
        \foreach \v [count=\i from 0] in {0,3,8,11,16,19}{\pgfmathsetmacro{\j}{int(Mod(\i,2)+1)}\node[txtnode,right=\dcl of \v]{$(v_\j,\i)$};}
        \node[txtnode,above left=-.05 and -.5 of 4]{$(v_3,2)$};
        \node[txtnode,above left=-.05 and -.5 of 12]{$(v_3,4)$};
        
        \draw[ddsedge, dashed]($(19)+(-3,1)$)--(19);
        \draw[ddsedge, dashed]($(19)+(-1,1)$)--(19);
        \draw[ddsedge, dashed]($(19)+(1,1)$)--(19);
        \draw[ddsedge, dashed]($(19)+(3,1)$)--(19);
    \end{scope}
    
    \begin{scope}[xshift=8.75cm,xscale=.1,yscale=.8]
        \foreach \x/\y [count=\v from 0] in {14.5/0,8.5/1,10.5/1,12.5/1,16.5/1,13.5/2,11.5/2,18.5/2,9.5/3,13.5/3,15.5/3,17.5/3,21.5/3,20.5/4,18.5/4,23.5/4,16.5/5,20.5/5,22.5/5,24.5/5,26.5/5}
        {\node[cnode](\v)at(\x,\y){$\bullet$};}
        \foreach \u/\v in {1/0,2/0,3/0,4/0,5/4,6/4,7/4,8/6,9/7,10/7,11/7,12/7,13/12,14/12,15/12,16/14,17/15,18/15,19/15,20/15}
        {\draw[ddsedge](\u)--(\v);}
        \foreach \v [count=\i from 0] in {0,4,7,12,15,20}{\pgfmathsetmacro{\j}{int(Mod(\i+1,2)+1)}\node[txtnode,right=\dcl of \v]{$(v_\j,\i)$};}
        \node[txtnode,above left=-.05 and -.5 of 8]{$(v_3,3)$};
        \node[txtnode,above left=-.05 and -.5 of 16]{$(v_3,5)$};
        
        \draw[ddsedge, dashed]($(20)+(-2,1)$)--(20);
        \draw[ddsedge, dashed]($(20)+(0,1)$)--(20);
        \draw[ddsedge, dashed]($(20)+(2,1)$)--(20);
    \end{scope}
    \end{tikzpicture}
    \vspace{-1.5em}
    \caption{The unroll $\unroll{A}$ of a disconnected FDDS $A$. Only the first $6$ levels of $\unroll{A}$ are shown. Both the FDDS and its unroll have depth $2$.}
    \label{fig:unroll}
\end{figure}

	We now recall three operations defined in \cite{gadouleau2022factorisation_dds} that will be useful later.
	Given a forest $\forest{f}$,
	we denote by $\dt{\forest{F}}$ the multi-set of trees rooted in the predecessors of the roots of $\forest{F}$.
	Then, we denote by $\rf{\forest{f}}$ the tree such that $\dt{\rf{\forest{f}}} = \forest{F}$. 
	Intuitively, this second operation connects the trees to a new common root.
	Finally, given a positive integer $k$, we denote $\cut{\tree{t}}{k}$ the induced sub-tree of $\tree{t}$ composed by the vertices with a depth less or equal to $k$.
	Let us generalize the same operation applied to a forest $\forest{f} = \tree{t}_{1} + ... + \tree{t}_{n}$ as  $\cut{\forest{F}}{k} = \cut{\tree{t}_{1}}{k} + ... + \cut{\tree{t}_{n}}{k}$.

		\section{Complexity of FDDS division with connected quotient}\label{section:division_FDDS}
	
	In this section we establish an upper bound to the complexity of division over FDDS. 
    More formally,
    our problem is to decide if, given two FDDS $A$ and $B$, there exists a connected FDDS $X$ such that $AX=B$.
    To achieve this, we will initially prove that cancellation holds over unrolls, \ie that $\forest{E}\forest{X} = \forest{E}\forest{Y}$ implies $\forest{X} = \forest{Y}$ for unrolls $\forest{E},\forest{X},\forest{Y}$.
    Later, we will extend the algorithm proposed in \cite[Figure 6] {gadouleau2022factorisation_dds} to handle more general unrolls (rather than just those consisting of a single tree), ultimately leading to our result.
    
    We begin by considering the case of forests containing a finite number of finite trees; we will refer to them as \emph{finite tree forests}. We will later generalise the reasoning to forests such as unrolls.
	
	\begin{lemma}\label{lemme:forest2tree}
		Let $\forest{A}$, $\forest{X}$, and $\forest{B}$ be finite tree forests. Then, $\forest{A}\forest{X} = \forest{B}$ if and only if $\rf{\forest{A}}\rf{\forest{x}} = \rf{\forest{B}}$.
	\end{lemma}
	
	\begin{proof}
		$(\Leftarrow)$ Assume $\rf{\forest{A}}\rf{\forest{X}} = \rf{\forest{b}}$.
		Then,$\dt{\rf{\forest{A}}\rf{\forest{x}}} = \dt{\rf{\forest{b}}} = \forest{B}$.
		Moreover, since $\rf{\forest{A}}$ and $\rf{\forest{X}}$ are finite trees, by \cite[Lemma 7]{gadouleau2022factorisation_dds} we have:
		\[\dt{\rf{\forest{A}}\rf{\forest{x}}} = \dt{\rf{\forest{A}}}\dt{\rf{\forest{x}}} = \forest{A} \forest{X}.\]
		$(\Rightarrow)$ We can show the other direction by a similar reasoning.
	\end{proof}
	
	Thanks to this lemma, we can generalise Lemma 21 of \cite{gadouleau2022factorisation_dds} as follows.
	
	\begin{lemma}\label{lemme:cancelFiniForest}
		Let $\forest{A}$, $\forest{X}$, and $\forest{Y}$ be finite tree forests. Then $\forest{A}\forest{X} = \forest{A} \forest{Y}$ if and only if $\cut{\forest{X}}{\depth{\forest{A}}} = \cut{\forest{Y}}{\depth{\forest{A}}}$.
	\end{lemma}
	
	\begin{proof}
		$(\Leftarrow)$ 
	    By the definition of tree product,
	    all nodes of $\forest{x}$ (resp., $\forest{y}$) of depth larger than $\depth{\forest{a}}$ do not impact the product $\forest{A} \forest{X}$ (resp., $\forest{A} \forest{Y}$). 
		Thus, we have $\forest{A} \forest{X} = \forest{A}\cut{\forest{X}}{\depth{\forest{A}}}$ and $\forest{A} \forest{Y} = \forest{A}\cut{\forest{Y}}{\depth{\forest{A}}}$.
		Since $\cut{\forest{X}}{\depth{\forest{A}}} = \cut{\forest{Y}}{\depth{\forest{A}}}$, we conclude that $\forest{A} \forest{X} = \forest{A} \forest{Y}.$
		
		$(\Rightarrow)$ Suppose $\forest{A}\forest{X} = \forest{A}\forest{Y}$.
		By Lemma \ref{lemme:forest2tree}, we have $\rf{\forest{A}}\rf{\forest{x}} = \rf{\forest{A}} \rf{\forest{y}}$.
		Since  $\rf{\forest{A}}$, $\rf{\forest{X}}$, and $\rf{\forest{y}}$ are finite trees, we deduce \cite[Lemma 21]{gadouleau2022factorisation_dds} 
		\begin{equation}\label{eq:cutx_equals_cuty}
	    \cut{\rf{\forest{x}}}{depth(\rf{\forest{a}})} = \cut{\rf{\forest{y}}}{depth(\rf{\forest{a}})}
		\end{equation}
		For all forest $\forest{F}$ and $d>0$, we have that $\mathcal{D}(\cut{\forest{F}}{d})$ is the multiset containing the subtrees rooted on the predecessors of the roots of  $\cut{\forest{F}}{d}$.
		It is therefore the same multiset as that which is composed of the subtrees rooted on the predecessors of the roots of $\forest{F}$ cut at depth $d-1$. 
		It follows that 
		$\cut{\dt{\forest{F}}}{d-1} = \dt{\cut{\forest{F}}{d}}$.
		In particular, for $\forest{F} = \rf{\forest{X}}$ and $d = \depth{\forest{A}} + 1 = \depth{\rf{\forest{A}}}$, we have \[\dt{\cut{\rf{\forest{X}}}{\depth{\forest{\rf{A}}}}} = \cut{\forest{X}}{\depth{\forest{A}}}.\]
		Likewise, $\dt{\cut{\rf{\forest{y}}} {\depth{\rf{\forest{a}}}}} = \cut{\forest{Y}}{\depth{\forest{A}}}$.
		By applying $\dt{\cdot}$ to both sides of \eqref{eq:cutx_equals_cuty}, we conclude $\cut{\forest{X}}{\depth{\forest{A}}} = \cut{\forest{Y}}{\depth{\forest{A}}}$.
	\end{proof}
	
	Lemma \ref{lemme:cancelFiniForest} is a sort of cancellation property subject to a depth condition.
	The first step to prove cancellation over unrolls is proving the equivalence between the notion of divisibility of unrolls and divisibility over deep enough finite cuts, as explained in the following proposition. 
	
	\begin{prop}\label{prop:divideForestFini}
		Let $A, X$ and $B$ be two FDDS with $\alpha$ equal to the number of unroll trees of $\unroll{B}$.
	    Let $n \ge 2 \, \alpha + \depth{\unroll{B}}$. Then
		\begin{center}
		    $\unroll{A} \unroll{X} = \unroll{B}$ if and only if $\cutUn{A} \cutUn{X} = \cutUn{B}$.
		\end{center}
	\end{prop}
	In the proof of Proposition~\ref{prop:divideForestFini}, we will see that such a depth is big enough
        to distinguish different unroll trees. 
	Before that, we introduce the notion of \emph{periodic pattern} of an unroll tree.
	Recall that an unroll tree $\tree{t}$ has exactly one infinite branch on which the trees $(\tree{t}_0,\tree{t}_1,\ldots)$ representing transient behaviour hook and repeat periodically.
	Let $p$ be a positive integer.
	A periodic pattern with period $p$ of $\tree{t}$ is a sequence of $p$ finite trees $(\tree{t}_0,\ldots,\tree{t}_{p-1})$ rooted on the infinite branch such that, for all $i\in \mathbb{N}\xspace$ we have $\tree{t}_i=\tree{t}_{i \bmod p}$.
	Let us point out that the idea here is to obtain a set of trees such that we represent all different behaviours repeating in all unroll trees, obtaining a finite representation.

    For connected FDDS, since its period $p$ is the number of trees in its unroll, we can reconstruct the FDDS itself from a periodic pattern $(\tree{t}_0,\ldots,\tree{t}_{p-1})$ of one of its unroll trees $\tree{t}_u$ by adding edges between $\tree{t}_i$ and $\tree{t}_{(i+1) \bmod p}$ for all $i$. We call this operation the \emph{roll of $\tree{t}_u$ of period $p$}. The following lemma shows that we can recover the periodic pattern of an unroll tree from a deep enough cut.
 	
	\begin{lemma}\label{lemme:RecoverPeridicsPart}
	    Let $A$ be a connected FDDS of period $p$, $\tree{t}$ be an unroll tree of $\unroll{A}$, and $n \ge p  + \depth{\unroll{A}}$.
		Let $(v_n,\ldots,v_0)$ be a directed path in $\cut{\tree{t}}{n}$ such that $\depth{v_n}=n$ and $v_0$ is the root of the tree.
		Then, nodes $v_p,\ldots,v_0$ necessarily come from the infinite branch of $\tree{t}$.
	\end{lemma}
	
	\begin{proof}
		We assume, by contradiction, that at least one of the nodes $v_p,\ldots,v_0$ does not come from the infinite branch of $\tree{t}$.
		Let $v_a$ be the node of $(v_{p-1},\ldots,v_0)$ with maximal depth coming from the infinite branch of $\tree{t}$; there always is at least one of them, namely the root $v_0$.		
		We have $\depth{v_n}\leq \depth{v_a}+\depth{\tree{t}}$. However, we assumed $\depth{v_a}<p$, thus $\depth{v_n}<p+\depth{\tree{t}}$. Since, $\depth{v_n}=n$, we have $n<p+\depth{\tree{t}}=p+\depth{\unroll{A}}$ which is a contradiction.
	\end{proof}
	
	Now, we can prove the following properties.
        Observe that an unroll tree $\tree{t}$ has multiple periods.
        In particular, if $p$ is a period of $\tree{t}$, then $k\,p$ is also a period of $\tree{t}$
        for all positive integer $k$.
	
	\begin{lemma}\label{lemma:depth_for_dif}
	    Let $\tree{a}$ and $\tree{b}$ two unroll trees with period $a$ and $b$ respectivly.
	    Let $n \ge 2\,\max(a,b) + \max(\depth{\tree{a}}, \depth{\tree{b}})$.
	    Then
		\begin{center}
		    $\tree{a} \neq \tree{b}$ if and only if $\cut{\tree{a}}{n} \neq \cut{\tree{b}}{n}$.
		\end{center}
	\end{lemma}
	
	\begin{proof}
          For the sake of this proof (especially $\Rightarrow$),
          we will assume that $a$ and $b$ are the smallest period,
          but the result immediatly follows without this restriction.
	    $(\Leftarrow)$
            This direction is very intuitive:
            if $\tree{a}$ and $\tree{b}$ are different when they are cut at depth $n$,
            then it is not possible that the additional nodes present in $\tree{a}$ and $\tree{b}$
            (but absent in $\cut{\tree{a}}{n}$ and $\cut{\tree{b}}{n}$) make them equal.
            We present a formal proof based on the total order on trees.
            We denote by $\codeT{\tree{x}}$ the code of $\tree{x}$ used to define the order of \cite{gadouleau2022factorisation_dds},
            and by $\codeTi{\tree{x}}$ the $i$-th integer in the code. 
           
	    We assume that $\cut{\tree{a}}{n} \neq \cut{\tree{b}}{n}$. 
	    Without loss of generality, we consider $\cut{\tree{a}}{n} < \cut{\tree{b}}{n}$. 
	    So, there exists an index $i$ such that $\codeTi{\cut{\tree{a}}{n}} < \codeTi{\cut{\tree{b}}{n}}$. 
	    Let $min$ be the index of the first difference. 
	    We set $u_a$ (resp. $u_b$) the node corresponding to $min$ in $\cut{\tree{a}}{n}$ (resp. $\cut{\tree{b}}{n}$). 
	    First, $u_a$ and $u_b$ have the same depth. 
	    Indeed, if it is not the case, since $\codeTS{\cut{\tree{a}}{n}}{j} = \codeTS{\cut{\tree{b}}{n}}{j}$ for all $j < min$, the sum of the codes of node with depth $\min(\depth{u_a},\depth{u_b}) - 1$ is the same. 
	    Therefore, the number of nodes with depth $\min(\depth{u_a},\depth{u_b})$ is the same in $\cut{\tree{a}}{n}$ and $\cut{\tree{b}}{n}$. 
	    However, since $u_b$ and $u_a$ do not have the same depth, by the definition of the code, the number of nodes with depth $\min(\depth{u_a},\depth{u_b})$ is different between $\cut{\tree{a}}{n}$ and $\cut{\tree{b}}{n}$. 
	    In addition, $u_b$ cannot be a leaf. 
	    Indeed, $\codeTS{\cut{\tree{b}}{n}}{min} > \codeTS{\cut{\tree{a}}{n}}{min} \ge 0$. 
	    Thus, $\depth{u_b} = \depth{u_a} < n$. 
	    We conclude that the difference between $\cut{\tree{a}}{n}$ and $\cut{\tree{b}}{n}$ is also a difference between $\tree{a}$ and $\tree{b}$.
	    
	    ($\Rightarrow$) We assume that $\tree{a} \neq \tree{b}$. 
	    We set $\tree{a}_i$ (resp. $\tree{b}_i$) the $i$-th finite tree rooted in the infinite branche of $\tree{a}$ (resp. $\tree{b}$).
	    By contradiction, we assume that $\cut{\tree{a}}{n} = \cut{\tree{b}}{n}$. 
	    Two cases are possible. 
	    First, $a = b$. 
	    So, since $n > a + \max(\depth{\tree{a}},\depth{\tree{b}})$, by Lemma~\ref{lemme:RecoverPeridicsPart}, we recover $\tree{a}_i$ and $\tree{b}_i$ for all $i$. 
	    Furthermore, we can also identify until $a$-th node come from the infinite branch in $\cut{\tree{a}}{n}$ and $\cut{\tree{b}}{n}$. 
	    We set $\tree{c}_i$ (resp. $\tree{d}_i$) the tree rooted in the $i$-th node of the cut of the infinite branch in $\cut{\tree{a}}{n}$ (resp. $\cut{\tree{b}}{n}$).
	    Since $\cut{\tree{a}}{n} = \cut{\tree{b}}{n}$, we have that $\dt{\tree{c}_0} = \dt{\tree{d}_0}$. 
	    Additionnaly, $\dt{\tree{c}_0}$ and $\dt{\tree{d}_0}$ contain one and only one tree with depth $n-1$, namely $\tree{c}_1$ and $\tree{d}_1$ repectively. 
        It follows that $\tree{a}_0 = \rf{\dt{\tree{c}_0} - \tree{c}_1} = \rf{\dt{\tree{d}_0} - \tree{d}_1} = \tree{b}_0$ (see Figure~\ref{fig:recover_PP}).
        \begin{figure}[tb]
        	\centering
        	\begin{tikzpicture}
	
	\coordinate (c0) at (0, 0);
	\coordinate (c1) at (1, 1);
	\coordinate (c2) at (2, 2);
	\coordinate (c3) at (3, 3);
	\coordinate (c4) at (4, 4);
	\coordinate (c5) at (5, 5);
	\coordinate (c6) at (5.2, 5.2);
	
	\draw[dashed, thick] (c0) -- ++(0.5, 1.5) -- ++(-1, 0) -- cycle;
	\draw[thick] (c1) -- ++(0.5, 1.5) -- ++(-1, 0) -- cycle;
	\draw[thick] (c2) -- ++(0.5, 1.5) -- ++(-1, 0) -- cycle;
	\draw[thick] (c3) -- ++(0.5, 1.5) -- ++(-1, 0) -- cycle;
	\draw[thick] (c4) -- ++(0.5, 1.2) -- ++(-1, 0) -- cycle;
	\draw[thick] (c5) -- ++(0, 0.2)   -- ++(-0.2, 0) -- cycle;
	
	\draw[thick] (c0) -- (c1) -- (c2) -- (c3) -- (c4) -- (c5) -- (c6);
	
	\node[below left] at (c0) {$c_0$};
	\node[below right] at (c1) {$c_1$};
	\node[below right] at (c2) {$c_2$};
	\node[below right] at (c3) {$c_3$};
	\node[below right] at (c4) {$c_4$};
	\node[below right] at (c5) {$c_5$};
	\node[below right] at (c6) {$c_6$};
	
	\node[below left] at (-0.5, 1) {$\mathcal{R}(\mathcal{D}(c_0) - c_1)$};
	
	\draw[dashed] (-2,5.2) -- (6,5.2);
	
\end{tikzpicture}
        	\caption{Example tree cut at depth 6.
                The dotted triangle illustrates $\mathcal{R}(\mathcal{D}(\tree{c}_0) - \tree{c}_1)$,
                in particular the usage of $-$ to remove the other part of the graph ($\tree{c}_1$).}
                \label{fig:recover_PP}
        \end{figure}
        
	    We can apply the same reasonnig for all $\tree{a}_i$ and $\tree{b}_i$ for $i$ between $1$ and $a - 1$. 
	    So, $\tree{a}$ and $\tree{b}$ have the same periodic pattern. 
	    This implies that $\tree{a} = \tree{b}$.
	    
	    Second, $a \neq b$. 
	    Without loss of generality, we assume that $a < b$. 
	    Since $n \ge 2 \, b + \max(\depth{\tree{a}}, \depth{\tree{b}})$,
            thanks to a similiar reasoning as the previous point, we deduce that $\tree{a}_{i \mod a} = \tree{b}_{i \mod b}$ for all $i$ between $0$ and $2 \, b - 1$
            (the factor $2$ is important here, see Figure~\ref{fig:dif_mult_2}).
            \begin{figure}[tb]
            	\begin{align*}
            		&\tree{a}_0 = \tree{b}_0 = \tree{b}_3&  &\tree{a}_0 = \tree{b}_0 = \tree{b}_3 =\tree{b}_2 = \tree{a}_2 \\
            		&\tree{a}_1 = \phantom{\tree{b}_0 ={}} \tree{b}_1& &\tree{a}_1 = \tree{b}_1 = \tree{b}_0 \phantom{\tree{b}_0 ={}}= \tree{a}_0\\
            		&\tree{a}_2 = \phantom{\tree{b}_0 ={}} \tree{b}_2&  &\tree{a}_2 = \tree{b}_2 = \tree{b}_1 \phantom{\tree{b}_0 ={}}= \tree{a}_1 \\[-1.2cm]
            	\end{align*} 
            	\caption{On the left, the case without the multiplication by a factor two in the lower bound for $n$,
                and on the right, the case with the multiplication by a factor two.
                One can see that without the factor two, it is not possible to continue
                the proof by identifying $\tree{a}_i$ with some $\tree{a}_j$.}
                \label{fig:dif_mult_2}
            \end{figure}
            
	    Thus, for each $i$ between $0$ and $a-1$, we have that $\tree{a}_i = \tree{b}_i = \tree{b}_{(b + i) \mod b} = \tree{a}_{(b + i) \mod a}$. 
	    So, either $i = (b + i) \mod a$ and there exists $k$ such that $b + i = a k + i$. 
	    This implies that $a$ is a divisor of $b$. 
	    In this case, $\tree{a}_i = \tree{b}_{i + aj}$ for all integer $i < a$ and $j < k$. 
	    It follow that $\tree{b}_i = \tree{b}_{i + a} = \ldots = \tree{b}_{i + a (k -1)}$ for all $i < a$. 
	    We conclude that $\tree{b}$ also has period $a$,
	    contradicting the minimality of $b$. 
	    
	    Or, $i \neq (b + i) \mod a$. 
	    Therefore, $\tree{a}_i = \tree{a}_j$ with $j = (b + i) \mod a$. 
	    Thereby, $\tree{a}_i = \tree{a}_{(b + i) \mod a} = \tree{a}_{(((b+i) \mod a) + b )\mod a} = \ldots$. 
	    However, $(((b + i) \mod a) + b) \mod a = (((b+i) \mod a) + (b \mod a)) \mod a = (2b + i) \mod a$. 
	    We can extend this process and we obtaint that $\tree{a}_i = \tree{a}_{(k b + i) \mod a}$ for all postive integer $k$ and all $i$ between $0$ and $a-1$. 
	    Thus, since:
            \[
              (kb + i) \mod a = ((k - \cfrac{\lcm(a,b)}{b})b + i) \mod a
            \]
            for all positive integer $k \ge \cfrac{\lcm(a,b)}{b}$, we deduce:
            \[
              |C_i = \{(k b + i) \mod a \mid k \in \N\}| = \cfrac{\lcm(a,b)}{b}
            \]
            for each integer $i$ between $0$ and $a-1$.
	    Also, if we suppose that $C_i$ is sorted by ascending order then the difference between two consecutive elements is $\gcd(a,b)$.
	    Consequently, we have that $\tree{a}_j = \tree{a}_k$ for all integer $i$ between $0$ and $a-1$ and $j,k \in C_i$. 
	    Finally, we remark that $C_0 \cup \ldots \cup C_{\gcd(a,b) - 1} = \{0, \ldots, a - 1\}$. 
	    This allows us to conclude that $\tree{a}$ also has period $\gcd(a,b)$. 
	    And since $a$ is not a divisor of $b$, it follows that $\gcd(a,b) < a$,
	    contradicting the minimality of $a$. 
   \end{proof}
   
   \begin{lemma}\label{lemma:equlaity_of_depth}
   		Let $A$ and $B$ be two FDDS with $\alpha$ equal to the number of unroll trees of $\unroll{B}$.
   		Let $n \ge 2 \, \alpha + \depth{\unroll{B}}$. 
   		If $\cutUn{A} = \cutUn{B}$ then $\depth{A} = \depth{B}$.
   \end{lemma}
   
   \begin{proof}
   		We prove the contraposition, assuming $\depth{A} \neq \depth{B}$. 
   		Without loss of generality, let $\depth{B} < \depth{A}$. 
   		Since $n$ is large enough, we deduce that for each $\tree{t}_b \in \cutUn{B}$,
                $\dt{\tree{t}_b}$ contain exactly one tree of depth $n-1$ and some tree of depth at most $\depth{B}-1$. 
   		However, there exists a tree $\tree{t}_a \in \cutUn{A}$ such that $\dt{\tree{t}_a}$ contains one tree of depth $n-1$ and one tree of depth $\depth{A} - 1$.
   		Therefore, since  $\depth{B}-1 < \depth{A}-1$, we deduce that no tree in $\cutUn{B}$ is equal to $\tree{t}_a$,
   		and conclude that $\cutUn{A} \neq \cutUn{B}$
   \end{proof}
   
   \begin{corollary}\label{cor:equlaity_of_depth}
   	 	Let $A,X$ and $B$ be FDDS with $\alpha$ equal to the number of unroll trees of $\unroll{B}$.
   	 	Let $n \ge 2 \, \alpha + \depth{\unroll{B}}$. 
   	 	If $\cutUn{A} \cutUn{X} = \cutUn{B}$ then $\depth{A} \le \depth{B}$ and $\depth{X} \le \depth{B}$.
   \end{corollary}
   
   \begin{proof}
   		We assume that $\cutUn{A} \cutUn{X} = \cutUn{B}$. 
   		Thus, $\cutUn{AX} = \cutUn{B}$.
   		Therefore, by Lemma~\ref{lemma:equlaity_of_depth}, we have that $\depth{AX} = \depth{B}$. 
   		And since $\depth{AX} = \max(\depth{A}, \depth{X})$, the result follows.
   \end{proof}

   \begin{lemma}\label{lemma:equality}
	   	Let $A$ and $B$ be two FDDS with $\alpha$ equal to the number of unroll trees of $\unroll{B}$.
	   	Let $n \ge 2 \, \alpha + \depth{\unroll{B}}$. Then
	   	\begin{center}
	   		$\unroll{A} = \unroll{B}$ if and only if $\cut{\unroll{A}}{n} = \cut{\unroll{B}}{n}$
	   	\end{center}
   \end{lemma}

   \begin{proof}
       $(\Rightarrow)$ If $\unroll{A} = \unroll{B}$ then $\cutU{A}{m} = \cutU{B}{m}$ for all positive integer $m$. 
       
        $(\Leftarrow)$ We assume that $\cutUn{A} = \cutUn{B}$.
		First, we explain how we recover $\unroll{A}$ and $\unroll{B}$ from $\cutUn{A}$ and $\cutUn{B}$ respectively. 
		We consider that $\unroll{A} = \sum_{i=1}^{\alpha_a} \tree{a}_i$
                and $\unroll{B} = \sum_{i=1}^{\alpha} \tree{b}_i$ and we set that $a_i$ (resp. $b_i$)
                be the length of a periodic pattern of $\tree{a}_i$ (resp. $\tree{b}_i$) with minimal length. 
        From the definition of unroll, we remark that $\max_i(a_i) \le \alpha_a$ and $\max_i(b_i) \le \alpha$. 
        Besides, by the hypothesis, it follows that $\alpha_a \le \alpha$.
        Thus, for each tree $\tree{t}$ of $\unroll{A}$ and $\unroll{B}$, there exists a periodic pattern of size $p$ of $\tree{t}$ such that $p \le \alpha$.
    
        Additionnaly, by Lemma~\ref{lemma:equlaity_of_depth}, we have that $\depth{A} = \depth{B}$.
        Therefore, $\depth{\unroll{A}} \le \depth{\unroll{B}}$,
        and we deduce that $n \ge 2 \max_{i,j}(a_i,b_j) + \max(\depth{\unroll{A}}, \depth{\unroll{B}})$. 
    
        Thereby, for $\tree{t}$ a tree of $\cutUn{A}$ and $p\le \alpha$ the size of one of its periodic pattern,
        we can extend $\tree{t}$ as an unroll tree $\tree{t'}$ such that $\cutTn{t'} = \tree{t}$.
        Moreover, each node of $\tree{t'}$ is of the form $(u,h)$ with $h \ge n$,
        and has the node $(v,h+1)$ as predecessor if and only if $(v, h + 1 - p)$ is a predecessor of $(u, h - p)$.
        Indeed, this operation is well defined by the two previous points. 
		We apply the same method for the trees in $\cutUn{B}$. 
		Also, $n \ge 2 \, p + \depth{\tree{t}}$.
		
		Let $i < \alpha_a$.
		From the hypothesis, there exists a tree $\tree{b}$ of $\cutUn{B}$ such that $\cut{\tree{a}_i}{n} = \tree{b}$.
		We set $\tree{t}_b$ the unroll tree obtain from $\tree{b}$ by the previous method. 
		Therefore, we have that $\cut{\tree{a}_i}{n} = \tree{b} = \cut{\tree{t}_b}{n}$. 
		Now, since $a_i \le \alpha_a \le \alpha$ and $\depth{\tree{a}_i} \le \depth{\unroll{A}} \le \depth{\unroll{B}}$,
                we deduce that $n \ge 2 \, a_i + \depth{\unroll{A}}$. 
		Hence, by Lemma~\ref{lemma:depth_for_dif}, we conclude that $\tree{t}_b = \tree{a_i}$,
		and that $\unroll{A} = \unroll{B}$.
   \end{proof}
	
	\begin{proof}[Proposition~\ref{prop:divideForestFini}]
		From Lemma~\ref{lemma:equality}, we have that $\unroll{A} \unroll{X} = \unroll{AX} = \unroll{B}$ if and only if $\cutUn{AX} = \cutUn{B}$. 
		However, since $\cutUn{AX} = \cutUn{A} \cutUn{X}$,
                we conclude that $\cutUn{A} \cutUn{X} = \cutUn{B}$ and the proposition follows. 
	\end{proof}
	
    Remark that the cut operation over $\unroll{B}$ at a depth $n$ generates a forest where the size of each tree is in $\bigo{m^2}$ and the total size is in $\bigo{m^3}$
    with $m$ the size of $B$ (\ie its number of nodes), since the chosen $n$ is at most $m$. Now, we can prove the main result of this section.
	
	\begin{theorem}
		For unrolls $\forest{A}$, $\forest{X}$, $\forest{Y}$ we have $\forest{A} \forest{X} = \forest{A} \forest{Y}$ if and only if $\forest{X} = \forest{Y}$.
	\end{theorem} 
	
	\begin{proof}
		Let $\alpha$ be the number of trees in $\forest{AX}$ and $n \geq \alpha + \depth{\forest{ax}}$ be an integer.
		By Proposition \ref{prop:divideForestFini},
		$\forest{AX} = \forest{AY}$ if and only if $\cut{\forest{A}}{n} \cut{\forest{X}}{n} = \cut{\forest{A}}{n} \cut{\forest{y}}{n}$.
		In addition, by Lemma \ref{lemme:cancelFiniForest}, 
		$\cut{\forest{A}}{n} \cut{\forest{x}}{n} = \cut{\forest{a}}{n} \cut{\forest{y}}{n}$ if and only if $\cut{\forest{x}}{n} = \cut{\forest{y}}{n}$.
		By Proposition	\ref{prop:divideForestFini}, the theorem follows.
	\end{proof}

    We can finally describe a division algorithm for FDDS working under the hypothesis that the quotient is connected.	

	\begin{algo}\label{algo:unrollDiv}
		Given two FDDS $A$ and $B$, where $\unroll{B}$ has $\alpha$ trees, we can compute $X$ such that $X$ is a connected FDDS and $A X = B$ (if any exists) by:
		\begin{enumerate}
			\item cutting $\unroll{A}$ and $\unroll{B}$ at depth $n = 2 \, \alpha + \depth{\unroll{B}}$;
			\item  computing $\tree{x}$ with the division algorithm \cite{gadouleau2022factorisation_dds}
                          to divide the tree $\rf{\cut{\unroll{B}}{n}}$ by $\rf{\cut{\unroll{A}}{n}}$\label{algo:unrollDiv:division}\label{algo:unrollDiv:computeDiv};
			\item computing the connected FDDS $X$ as the roll of period $p$ of any tree of $\dt{\tree{x}}$, where $p$ is equal to the number of trees in $\dt{\tree{x}}$;
			\item and verifying if $X$ multiplied by $A$ is isomorphic to $B$. \label{algo:unrollDiv:verification}
		\end{enumerate}
	\end{algo}
	
    Since the depth where we cut is large enough, Proposition~\ref{prop:divideForestFini}, Lemma~\ref{lemme:forest2tree} and the correctness of the division algorithm of \cite{gadouleau2022factorisation_dds} imply that the tree $\tree{x}$ computed in Step~\ref{algo:unrollDiv:computeDiv} of Algorithm~\ref{algo:unrollDiv} satisfies $\cut{\unroll{A}}{n}\dt{\tree{x}} = \cut{\unroll{B}}{n}$. 
    By the definition of unroll, since we only search for connected FDDS, if $\dt{\tree{x}}$ is the cut of an unroll then the rolls of each tree of $\dt{\tree{x}}$ at period $p$ are isomorphic. Furthermore, Lemma~\ref{lemme:RecoverPeridicsPart} ensures that we can roll each tree in $\dt{\tree{x}}$.
    However, $\dt{\tree{x}}$ is not necessarily the cut of an unroll and it is possible that there exists an FDDS $X$ such that $\dt{\tree{x}} = \cut{\unroll{X}}{n}$ but $AX \neq B$ (an example can be seen in Figure~\ref{fig:unroll-counter-example}). 
    As a consequence, Step~\ref{algo:unrollDiv:verification} of Algorithm~\ref{algo:unrollDiv} is mandatory to ensure its correctness.

	\begin{figure}[tb]
\centering
\begin{tikzpicture}[xscale=.8,yscale=.6]
\tikzset{ddsnode/.style={inner sep=0,outer sep=2}}
\draw[line width=.1mm,gray](-2,4.5)--(11,4.5);
\node[scale=1.9,fill=white]at(0,4.5){$A$};
\node[scale=1.9,fill=white]at(2,4.5){$B$};
\node[scale=1.9,fill=white]at(7,4.5){$C$};
\begin{scope}[yshift=5.5cm]
\begin{oodgraph}
\addcycle[yshift=1cm,rotation angle=130,nodes prefix = a]{2};
\addbeard[attach node = a->1,rotation angle=-80]{1};
\addbeard[attach node = a->2,rotation angle=-80]{1};

\addcycle[xshift=2.25cm,yshift=.25cm,rotation angle=90,nodes prefix = b]{1};
\addbeard[attach node = b->1,rotation angle=45]{1};
\addbeard[attach node = b->1->1,rotation angle=-90]{1};

\addcycle[xshift=5.5cm,yshift=0cm,nodes prefix = c,rotation angle=90]{1};
\addbeard[attach node = c->1,rotation angle=3]{3};
\addbeard[attach node = c->1->2]{2};
\addcycle[xshift=8.5cm,yshift=0cm,nodes prefix = d,rotation angle=90]{1};
\addbeard[attach node = d->1,rotation angle=3]{3};
\addbeard[attach node = d->1->2]{2};
\end{oodgraph}

\node[scale=2] at (1,1){$\times$};
\node[scale=2] at (3.25,1){$\neq$};
\end{scope}

\begin{scope}[xshift=-1cm,xscale=.2,yscale=1]
\foreach \x/\y [count=\v from 0] in {0/0,-1/1,1/1,0/2,2/2,1/3,3/3}
{\node[ddsnode](\v)at(\x,\y){$\bullet$};}
\foreach \u/\v in {1/0,2/0,3/2,4/2,5/4,6/4}
{\draw[ddsedge](\u)--(\v);}
\draw[ddsedge,dashed](2,4)--(6);
\draw[ddsedge,dashed](4,4)--(6);
\end{scope}

\begin{scope}[xshift=0cm,xscale=.2,yscale=1]
\foreach \x/\y [count=\v from 0] in {0/0,-1/1,1/1,0/2,2/2,1/3,3/3}
{\node[ddsnode](\v)at(\x,\y){$\bullet$};}
\foreach \u/\v in {1/0,2/0,3/2,4/2,5/4,6/4}
{\draw[ddsedge](\u)--(\v);}
\draw[ddsedge,dashed](2,4)--(6);
\draw[ddsedge,dashed](4,4)--(6);
\end{scope}

\begin{scope}[xshift=2cm,xscale=.1,yscale=1]
\foreach \x/\y [count=\v from 0] in {0/0,-2/1,2/1,-2/2,0/2,4/2,0/3,2/3,6/3}
{\node[ddsnode](\v)at(\x,\y){$\bullet$};}
\foreach \u/\v in {1/0,2/0,3/1,4/2,5/2,6/4,7/5,8/5}
{\draw[ddsedge](\u)--(\v);}
\draw[ddsedge,dashed](4,4)--(8);
\draw[ddsedge,dashed](8,4)--(8);
\end{scope}

\begin{scope}[xshift=5cm,xscale=.2,yscale=1]
\foreach \x/\y [count=\v from 0] in {0/0,-3/1,-1/1,1/1,3/1,-4/2,-2/2,0/2,2/2,4/2,6/2,-1/3,1/3,3/3,5/3,7/3,9/3}
{\node[ddsnode](\v)at(\x,\y){$\bullet$};}
\foreach \u/\v in {1/0,2/0,3/0,4/0,5/1,6/1,7/4,8/4,9/4,10/4,11/7,12/7,13/10,14/10,15/10,16/10}
{\draw[ddsedge](\u)--(\v);}
\draw[ddsedge,dashed](6,4)--(16);
\draw[ddsedge,dashed](8,4)--(16);
\draw[ddsedge,dashed](10,4)--(16);
\draw[ddsedge,dashed](12,4)--(16);
\end{scope}

\begin{scope}[xshift=8cm,xscale=.2,yscale=1]
\foreach \x/\y [count=\v from 0] in {0/0,-3/1,-1/1,1/1,3/1,-4/2,-2/2,0/2,2/2,4/2,6/2,-1/3,1/3,3/3,5/3,7/3,9/3}
{\node[ddsnode](\v)at(\x,\y){$\bullet$};}
\foreach \u/\v in {1/0,2/0,3/0,4/0,5/1,6/1,7/4,8/4,9/4,10/4,11/7,12/7,13/10,14/10,15/10,16/10}
{\draw[ddsedge](\u)--(\v);}
\draw[ddsedge,dashed](6,4)--(16);
\draw[ddsedge,dashed](8,4)--(16);
\draw[ddsedge,dashed](10,4)--(16);
\draw[ddsedge,dashed](12,4)--(16);
\end{scope}

\node[scale=2] at (1,1.5){$\times$};
\node[scale=2] at (3.25,1.5){$=$};


\end{tikzpicture}
\caption{Three FDDSs $A$, $B$ and $C$ such that $AB\neq C$ but $\unroll{A}\unroll{B}=\unroll{C}$. Here the symbol $\times$ denotes the product of FDDSs on the top, and of forests on the bottom.}
\label{fig:unroll-counter-example}
\end{figure}

	\begin{theorem} \label{th:compUnrollDiv}
		Algorithm \ref{algo:unrollDiv} runs in $\bigo{m^9}$ time, where $m$ is the size of its inputs.
	\end{theorem}

	\begin{proof}
	    The cuts of depth $n$ of the unrolls of $A$ and $B$ can be computed in $\bigo{m^3}$ time and the size of the result is $\bigo{m^3}$. 
	    In fact, we can construct $\cut{\unroll{A}}{n}$ and $\cut{\unroll{B}}{n}$ backwards from their roots up to depth $n$;  the size of $\cut{\unroll{A}}{n}$ is bounded by the size of $\cut{\unroll{B}}{n}$, which is $\bigo{m^3}$. 
		By analysing the division tree algorithm in Figure 6 of \cite{gadouleau2022factorisation_dds}, we can check that it can be executed in cubic time. 
		Moreover, since its inputs have size $\bigo{m^3}$, Step \ref{algo:unrollDiv:division} of Algorithm \ref{algo:unrollDiv} requires $\bigo{m^9}$ time.
	    The roll procedure of a tree can be computed by a traversal, requiring $\bigo{m^2}$ time.
	    Finally, the product of two FDDS is quadratic-time on its input but linear-time on its output.
	    However, in our case, the size of the output of $AX$ is bounded by the size of $B$; hence, the product can be computed in $\bigo{m}$ time. Finally, the isomorphism test requires $\bigo{m}$
            because our uniform outgoing degree one graphs are planar~\cite{hopcroft1974linear}.
	\end{proof}

		\section{Complexity of computing $k$-th roots of unrolls}\label{section:root_FDDS}

	The purpose of this section is to study the problem of computing connected roots on FDDS, particularly on transients.
	Let $\forest{A} = \tree{t}_1 + \ldots + \tree{t}_n$ be a forest and $k$ a positive integer.
	Then $\forest{A}^k = \sum_{k_1+\ldots+k_n = k} \binom{k}{k_1,\ldots,k_n} \prod_{i=1}^{n} \tree{t}_i^{k_i}$; furthermore, since the sum of forests is their disjoint union, each forest (in particular $\forest{A}^k$) can be written as a sum of trees in a unique way (up to reordering of the terms).
	The injectivity of $k$-th roots, in the semiring of unrolls, has been proved in \cite{gadouleau2022factorisation_dds}. 
	Here, we study this problem from an algorithmic and complexity point of view, and find a polynomial-time upper bound for the computation of $k$-th roots.
			
	We begin by studying the structure of a forest of finite trees raised to the $k$-th power.
	Indeed, if we suppose $\forest{X} = \tree{t}_1 + \ldots +\tree{t}_n$ with $\tree{t}_i \le \tree{t}_{i+1}$, we want to be able to identify the smallest tree of $\forest{X}^k$ from the product $\tree{t}_i \times \prod_{j=1}^{n} \tree{t}_j$.
	Moreover, we want to be able to identify it for all $\tree{t}_i$.
	
	Hereafter, we consider $\tree{a}^0$ to be equivalent to the simple oriented path with length equivalent to the depth of $\tree{a}$ (the same is true for forests).
	
	\begin{lemma}\label{lemme:ordreProdFin}
    Let $\forest{X}$ be a forest of the form $\forest{X}=\tree{t}_1+\ldots+\tree{t}_n$ (with $\tree{t}_i\leq\tree{t}_{i+1}$) and $k$ a positive integer. For any tree $\tree{t}_{i}$ of depth $d_i$ in $\forest{X}$, the smallest tree $\tree{t}_s$ of depth $d_i$ with factor $\tree{t}_{i}$ in $\forest{X}^k$ is isomorphic to $\tree{t}_m^{k-1}\tree{t}_i$, where $\tree{t}_m$ is the smallest tree of $\forest{X}$ with depth at least $d_i$.
	\end{lemma}
	
	\begin{proof}
	    Let us assume that the smallest tree $\tree{t}_s$ of depth $d_i$ with factor $\tree{t}_{i}$ in $\forest{X}^k$ is not isomorphic to $\tree{t}_m^{k-1}\tree{t}_i$.
		Two cases are possible. Either $\tree{t}_s$ contains a third factor other than $\tree{t}_m$ and $\tree{t}_i$, or it is of the form $\tree{t}_m^{k-k_i}\tree{t}_i^{k_i}$, with $k_i>1$.
		
		In the former case, let us suppose that there exists $a \in \{1,\ldots,i-1\}\setminus\{m\}$ and $k_a > 0$ such that $\tree{t}_a \neq \tree{t}_m$ and $\tree{t}_s$ is isomorphic to $\tree{t}_i^{k_i} \tree{t}_a^{k_a} \tree{t}_m^{k_m}$. 
		Remark that, according to \cite[Lemma 10]{gadouleau2022factorisation_dds}, the smallest tree of depth $d_i$ with factor $\tree{t}_{i}$ in $\forest{X}^k$ necessarily has all its factors of depth at least $d_i$.
		For this reason, we can assume $\depth{\tree{t}_a} \ge d_i$ without loss of generality.
		However, since $\tree{t}_m < \tree{t}_a$, we have $\tree{t}_m^{k_m + 1} \tree{t}_a^{k_a - 1} < \tree{t}_m^{k_m} \tree{t}_a^{k_a}$.
		Thus, we have that $\tree{t}_i^{k_i} \tree{t}_m^{k_m + 1} \tree{t}_a^{k_a - 1}  < \tree{t}_i^{k_i} \tree{t}_m^{k_m} \tree{t}_a^{k_a}$.
		This brings us into contradiction with the minimality of $\tree{t}_s$.
		
		In the second case, we assume that $\tree{t}_s$ is isomorphic to  $\tree{t}_m^{k-k_i} \tree{t}_i^{k_i}$ with $k_i > 1$.
		By hypothesis, we have $\tree{t}_i \ge \tree{t}_m$.
		If we consider the case of $\tree{t}_m < \tree{t}_i$, we have 
		$\tree{t}_m^{k-k_i + 1} \tree{t}_i^{k_i -1} < \tree{t}_m^{k-k_i} \tree{t}_i^{k_i}$.
		Once again, this is in contradiction with the minimality of $\tree{t}_s$.
		In the case of $\tree{t}_m = \tree{t}_i$, we have $\tree{t}_m^{k - k_i +1}\tree{t}_i^{k_i -1} = \tree{t}_m^{k - k_i}\tree{t}_i^{k_i}$. But we supposed $\tree{t}_s$ not isomorphic to $\tree{t}_m^{k-1} \tree{t}_i$.
		This concludes the proof.
		\end{proof}
	
    Before describing an algorithmic technique for computing roots over unrolls (\ie forests), we need a last technical lemma.
   
	\begin{lemma}\label{lemme:compoRacineK}
		Let $\tree{x}$ and $\tree{a}$ be two finite trees such that $\tree{x}^k = \tree{a}$, and $k$ a positive integer. Then, $\dt{\tree{a}} = \dt{\tree{x}}^k$.
	\end{lemma}
	
	\begin{proof}
		Since $\tree{x}$ is a tree, for all $i \le k$, $\tree{x}^i$ is also a tree.
		According to \cite[Lemma 7]{gadouleau2022factorisation_dds}, we have $\dt{\tree{a}} = \dt{\tree{x}^k} = \dt{\tree{x}}^k$.
	\end{proof}
	
	We now introduce an algorithmic procedure to compute the roots over forests based on an induction over decreasing depths in which, each time, we reconstruct part of the solution considering the smallest tree with at least a specific depth (according to Lemmas \ref{lemme:ordreProdFin} and~\ref{lemme:compoRacineK}).
	
	\begin{algorithm}[t]
\caption{\texttt{root}}\label{alg:root}
\begin{algorithmic}[1]
\Require $\forest{A}$ a forest, $k$ an integer
\If {$\forest{A}$ is a path} \label{line:if-path}
    \State \Return $\forest{A}$ \label{line:path}
\EndIf
\State $\forest{R} \gets \varnothing$
\State $\tree{t}_m \gets \varnothing$
\State $\forest{F} \gets \forest{A}$ \label{line:copy-forest}
\While{$\forest{F} \neq \varnothing$} \label{line:while}
    \State $\forest{F} \gets \forest{A}\setminus \forest{R}^k$. \label{line:setminus}
    \State $\tree{t}_s \gets \min\{\tree{t} \mid \tree{t}\in\forest{F},\depth{\tree{t}}=\depth{\forest{F}}\}$ \label{line:min}
    \If {$\tree{t}_m=\varnothing\;\textbf{or}\;\tree{t}_m^k>\tree{t}_s$} \label{line:root-condition}
        \State $\tree{t}_i \gets \rf{\texttt{root}(\dt{\tree{t}_s},k)}$ \label{line:recursion}
        \State $\tree{t}_m \gets \tree{t}_i$
    \Else
        \State $\tree{t}_i \gets \texttt{divide}(\tree{t}_s,\tree{t}_m^{k-1})$ \label{line:divide}
    \EndIf
    \If{$\tree{t}_i = \bot\;\textbf{or}\;(\forest{R}+\tree{t}_i)^k\nsubseteq\forest{A}$} \label{line:check}
        \State \Return $\bot$
    \EndIf 
    \State $\forest{R} \gets \forest{R} + \tree{t}_i$
\EndWhile \label{line:end-while}
\State \Return $\forest{R}$
\end{algorithmic}
\label{algo:kRootForest}
\end{algorithm}

	\begin{theorem}
	    Given a forest $\forest{a}$ and $k$ a strictly positive integer, we can compute $\forest{x}$ such that $\forest{x}^k = \forest{a}$ with Algorithm \ref{alg:root}.
	\end{theorem}

	In Algorithm \ref{alg:root}, the main idea is to extract iteratively the minimal tree among the tallest ones in $\forest{A}$ (\ie $\tree{t}_s$).
	This tree will be used to reconstruct one of the trees of $\forest{X}$ (\ie $\tree{t}_i$). This can be done in two ways according to two possible scenarios.
	In the first case, $\tree{t}_s$ is smaller than the smallest one already reconstructed (\ie $\tree{t}_m$) raised to the power $k$. If this is the case, we compute a new tree in $\forest{X}$ through a recursive call to our $\texttt{root}$ function.
	In the second case, the extracted tree is greater than $\tree{t}_m^k$. This means, by Lemma~\ref{lemme:ordreProdFin}, that it is a product of the smallest reconstructed one (\ie $\tree{t}_m$) and a new one (\ie $\tree{t}_i$). In this case, the latter can be computed by the $\texttt{divide}$ algorithm of~\cite{gadouleau2022factorisation_dds}.
	After the reconstruction of a tree $\tree{t}_i$ of $\forest{X}$, we remove from $\forest{A}$ all the trees obtainable from products of already computed trees in $\forest{X}$. This allows us to extract progressively shorter trees $\tree{t}_s$ from $\forest{A}$ and to compute consequently shorter trees of $\forest{X}$.
	When we remove all trees in $\forest{A}$ obtainable from trees $\tree{t}_i$ with depth at least $d_i$ in $\forest{X}$, this leaves us only trees with depth at most $d_i$. Since for each depth, the number of trees of this depth is finite, the algorithm necessarily halts.

	Let us consider an example.	In Figure~\ref{fig:root_intuition}, in order to compute the left side from the right one, the first tree considered is $\tree{t}_1^2$, the single tallest one. The latter can be used to compute $\tree{t}_1$ recursively. Next, the smallest one among the remaining ones is $\tree{t}_0^2$, which is smaller than $\tree{t}_1^2$. Thus, we can compute $\tree{t}_0$ again through recursion. Finally, the last tree extracted, after removing the trees with exclusively $\tree{t}_0$ and $\tree{t}_1$ as factors, is $\tree{t}_0\tree{t}_2$. Since this time $\tree{t}_0^2$ is smaller, we can get the third and final tree $\tree{t}_2$ by dividing it by the smallest computed tree yet.
	\vspace{-0.2cm}
	
\newcommand{\ft}[1]{\tree{t}_{#1}}
\begin{figure}[b]
\centering
\begin{tikzpicture}[xscale=1.15,yscale=.65]
\tikzset{txtnode/.style={scale=1},
ddsnode/.style={scale=1.1,outer sep=0,inner sep=2}}
\begin{scope}
\foreach \t [count=\i from 0] in {$\ft{0}$,$\ft{1}$,$\ft{2}$}
{
\node[txtnode]at(\i,-.7){\t};
\node[ddsnode](\i_0_0)at(\i,0){$\bullet$};
}
\foreach \t [count=\i from 0] in {$<$,$<$}
{
\node[txtnode]at(\i+.5,-.7){\t};
\node[txtnode]at(\i+.5,0){$+$};
}

\node[ddsnode](0_1_0)at(0,1){$\bullet$};
\draw[ddsedge](0_1_0)--(0_0_0);

\node[ddsnode](1_1_0)at(1,1){$\bullet$};
\node[ddsnode](1_2_0)at(1,2){$\bullet$};
\draw[ddsedge](1_2_0)--(1_1_0);
\draw[ddsedge](1_1_0)--(1_0_0);

\node[ddsnode](2_1_0)at(1.8,1){$\bullet$};
\node[ddsnode](2_1_1)at(2.2,1){$\bullet$};
\draw[ddsedge](2_1_0)--(2_0_0);
\draw[ddsedge](2_1_1)--(2_0_0);

\draw (-.3,0) to [leftp] (-.3,2);
\draw (2.3,0) to [rightp] (2.3,2);
\node[txtnode] at (2.5,2){$2$};
\end{scope}
\node[txtnode]at(3,0){$=$};
\begin{scope}[xshift=4cm]
\foreach \t [count=\i from 0] in {$\ft{0}^2$,$\ft{0}\ft{1}$,$\ft{1}^2$,$\ft{0}\ft{2}$,$\ft{1}\ft{2}$,$\ft{2}^2$}
{
\node[txtnode]at(\i,-.7){\t};
\node[ddsnode](\i_0_0)at(\i,0){$\bullet$};
}
\foreach \t [count=\i from 0] in {$\leq$,$<$,$<$,$\leq$,$<$}
{\node[txtnode]at(\i+.5,-.7){\t};}
\foreach \t [count=\i from 0] in {$+\;2\;\times$,$+$,$+\;2\;\times$,$+\;2\;\times$,$+$}
{\node[txtnode,anchor=south]at(\i+.5,-.35){\t};}

\node[ddsnode](0_1_0)at(0,1){$\bullet$};
\draw[ddsedge](0_1_0)--(0_0_0);

\node[ddsnode](1_1_0)at(1,1){$\bullet$};
\draw[ddsedge](1_1_0)--(1_0_0);

\node[ddsnode](2_1_0)at(2,1){$\bullet$};
\node[ddsnode](2_2_0)at(2,2){$\bullet$};
\draw[ddsedge](2_2_0)--(2_1_0);
\draw[ddsedge](2_1_0)--(2_0_0);

\node[ddsnode](3_1_0)at(2.8,1){$\bullet$};
\node[ddsnode](3_1_1)at(3.2,1){$\bullet$};
\draw[ddsedge](3_1_0)--(3_0_0);
\draw[ddsedge](3_1_1)--(3_0_0);

\node[ddsnode](4_1_0)at(3.8,1){$\bullet$};
\node[ddsnode](4_1_1)at(4.2,1){$\bullet$};
\draw[ddsedge](4_1_0)--(4_0_0);
\draw[ddsedge](4_1_1)--(4_0_0);

\node[ddsnode](5_1_0)at(4.7,1){$\bullet$};
\node[ddsnode](5_1_1)at(4.9,1){$\bullet$};
\node[ddsnode](5_1_2)at(5.1,1){$\bullet$};
\node[ddsnode](5_1_3)at(5.3,1){$\bullet$};
\draw[ddsedge](5_1_0)--(5_0_0);
\draw[ddsedge](5_1_1)--(5_0_0);
\draw[ddsedge](5_1_2)--(5_0_0);
\draw[ddsedge](5_1_3)--(5_0_0);
\end{scope}
\draw[line width=.1mm,gray](-.5,-.3)--(9.5,-.3);
\end{tikzpicture}
\caption{Order of the trees in the square of a forest.}
\label{fig:root_intuition}
\end{figure}

	\begin{theorem}\label{th:compKRoot}
		Algorithm \ref{alg:root} runs in $\bigo{m^3}$ time if $k$ is at most $\lfloor \log_2 m \rfloor$, where $m$ is the size of $\forest{a}$.
	\end{theorem}

	\begin{proof}
	    If $\forest{A}$ is a path, then the algorithm halts in linear time $\bigo{m}$ on line \ref{line:path}.
        
        Otherwise, there exists a level $i$ of $\forest{A}$ containing $\beta \ge 2$ nodes. In order to justify the upper bound on $k$, suppose $\forest{R}^k = \forest{A}$. Then, level $i$ of $\forest{R}$ contains $\sqrt[k]{\beta}$ nodes. The smallest integer greater than $1$ having a $k$-th root is $2^k$, thus $\beta \ge 2^k$. Since $\beta \le m$, we have $k \le \lfloor \log_2 m \rfloor$.
        
        Lines \ref{line:if-path}--\ref{line:copy-forest} take linear time $\bigo{m}$. The while loop of lines \ref{line:while}--\ref{line:end-while} is executed, in the worst case, once per tree of the $k$-th root $\forest{R}$, \ie a number of times equal to the $k$-th root of the number of trees in $\forest{A}$. Line \ref{line:while} takes $\bigo{1}$ time.
        The product of trees requires linear time in its output size. Consequently, $\forest{R}^k$ can be computed in $\bigo{m \log m}$ time.
        Moreover, since we can remove $\forest{R}^k$ from $\forest{A}$ in quadratic time, we deduce that line \ref{line:setminus} takes $\bigo{m^2}$ time. 
        Since the search of $\tree{t}_s$ consists of a simple traversal, we deduce that line \ref{line:min} takes $\bigo{m}$ time. Line \ref{line:root-condition} takes $\bigo{m \log m}$ time for computing $\tree{t}_m^k$. If no recursive call is made, line \ref{line:divide} is executed in time $\bigo{m_s^3}$, where $m_s$ is the size of $\tree{t}_s$. The runtime of lines \ref{line:check}--\ref{line:end-while} is dominated by line \ref{line:check}, which takes $\bigo{m^2}$ as line \ref{line:setminus}.
        
        Since each tree $\tree{t}_s$ in $\forest{A}$ is used at most once, we have $\sum m_s^3 \le m^3$; as a consequence, the most expensive lines of the algorithm (namely, \ref{line:setminus}, \ref{line:divide}, and \ref{line:check}) have a total runtime of $\bigo{m^3}$ across all iterations of the while loop.
        
        We still need to take into account the recursive calls of line \ref{line:recursion}. By taking once again into account the bound $\sum m_s^3 \le m^3$, the total runtime of these recursive calls is also $\bigo{m^3}$.
        We conclude that Algorithm \ref{alg:root} runs in time $\bigo{m^3}$.
    \end{proof}
	
	\begin{corollary}\label{cor:findk}
	Let $\forest{A}$ be a forest. Then, it is possible to decide in polynomial time if there exists a forest $\forest{X}$ and an integer $k>1$ such that $\forest{X}^k=\forest{A}$. 
	\end{corollary}
	
	\begin{proof}
	Since $k$ is bounded by the logarithm of the size of $\forest{a}$ and, according to Theorem \ref{th:compKRoot}, we can compute the root in $\bigo{m^3}$, we can test all $k$ (up to the bound) and check if there exists a $\forest{X}$ such that $\forest{x}^{k} = \forest{a}$ in $\bigo{m^3 \log m}$, where $m$ denotes once again the size of $\forest{a}$. 
	\end{proof}
	
	According to Corollary \ref{cor:findk}, we easily conclude that the corresponding enumeration problem of finding all solutions $\forest{X}$ (for all powers $k$) is in $\EnumP$, since the verification of a solution can be done in polynomial time and the size of a solution is polynomial in the size of the input. 
	Moreover, the problem is in the class $\DelayP$ since the time elapsed between the computation of one solution (for a certain $k$) and the next is polynomial.
	We refer the reader to \cite{strozecki2021enum_complexity} for more information about enumeration complexity classes.
		
	Now that we have a technique to compute the root of forests, let us think in terms of unrolls of FDDS.
	Consider an FDDS $A$ and its unroll $\forest{A}=\unroll{A}$.
	According to Proposition \ref{prop:divideForestFini}, we can compute the FDDS $X$ such that $A=X^k$ by considering the forest $\forest{F}=\cut{\forest{A}}{n}$ where $\alpha$ is the number of trees in $\forest{f}$ and $n=\alpha+\depth{\forest{A}}$. 
    Again, this depth allows us to ensure that all the transient dynamics of the dynamical system are represented in the different trees.
	Applying the root algorithm on $\forest{F}$, we obtain the result of the root as a forest of finite trees.
	However, this is just a candidate solution for the corresponding problem over the initial FDDS (for the same reasoning as in the case of division).
	In order to test the result of the root algorithm, as before,
        we realise the roll of one tree in the solution to period $p$,
        with $p$ as the number of trees in the result.
	Then to decide if $X$ is truly the $k$-th root of $A$, we verify if $X^k=A$ where $X$ is the result of the roll operation.
	That is possible because the algorithm is designed to study connected solutions.
	Indeed, the following holds.

    \begin{corollary}\label{cor:complDDS}
	Let $A$ be a FDDS, it is possible to decide if there exists a connected FDDS $X$ and an integer $k>1$ such that $X^k=A$ in polynomial time. 
	\end{corollary}
	
    By combining the division algorithm with the root algorithm, we are now able to study equations of the form $AX^k = B$.
	Given FDDS $A$ and $B$ and $k > 0$, we can first compute the result $\forest{Y}$ of the division of $\cut{\unroll{B}}{n}$ by  $\cut{\unroll{A}}{n}$, where $\alpha$ is the number of trees in $\unroll{B}$ and $n = \alpha + \depth{B}$.
	Then, we compute the $k$-th root $\forest{X}$ of $\forest{Y}$.
	After that, we make the roll of one tree of $\forest{x}$ in period $p$, with $p$ the number of trees in $\forest{X}$.
	Then, using the roll result $X$, we just need to verify if $A X^k=B$. 
	Once again, the solutions found by this method are only the connected ones, and further non-connected solutions are also possible.

	\section{Unroll Quotient}\label{section:divison_unroll}
    
    Despite the fact that we have constructed a polynomial-time algorithm
    for the division of finite trees and forests,
    we do not have one for the division of unrolls (an unroll being encoded by its FDDS).
    Indeed, Proposition~\ref{prop:divideForestFini} is a first step,
    but we do not know how to recover an unroll from the finite forest
    obtained by the division algorithm from~\cite{gadouleau2022factorisation_dds}:
    \[
      \cfrac{\cutUn{B}}{\cutUn{A}}.
    \]
    We will achieve this in the present section, but to begin,
    we are not even sure that this quotient
    is always the cut of an unroll.
    We start by proving that this first step is sound,
    in the sense that there always exists an unroll whose cut is this quotient
    (Lemma~\ref{lemma:quotient_unroll2}).

    For this purpose, we introduce two new notations. 
    Given $\tree{u}$ an unroll tree, we denote by $\shift{\tree{u}}{i}$
    its \emph{shift}, obtained by removing the nodes of the infinite branch of $\tree{u}$
    (and the finite trees hooked at them)
    whose depth is less than $i$.
    We denote by $\Shift{\tree{u}}$ the set of unroll trees generated from
    $p$ the size of the smallest periodic pattern of $\tree{u}$, formally
    $\Shift{\tree{u}} = \bigcup_{i=0}^{p} \shift{\tree{u}}{i}$.
    We call $\Shift{\tree{u}}$ the \emph{shifts} of $\tree{u}$,
    it basically allows to reconstruct all the unroll trees of a connected component
    (cycle in the dynamical system),
    from a single unroll tree.
    The first lemma is technical, and does not yet speaks about $\forest{x}$.

    \begin{lemma}\label{lemma:shift}
        Let $A,B$ be two FDDS, $\tree{a},\tree{b}$ be two unroll trees of $\unroll{A}, \unroll{B}$ respectively. 
        If there exists $\tree{x}$ an infinite tree such that $\tree{a} \tree{x} = \tree{b}$ and $\Shift{\tree{a}} \tree{x} \subseteq \unroll{B}$ then $\tree{x}$ is an unroll tree and $\tree{a} \Shift{\tree{x}} \subseteq \unroll{B}$.
    \end{lemma}
    
    \begin{proof}
        We assume that there exists $\tree{x}$ such that $\tree{a} \tree{x} = \tree{b}$ and $\Shift{\tree{a}}\tree{x} \subseteq \unroll{B}$.
        This implies that $\tree{x}$ is an unroll tree
        with a periodic pattern of size $\lcm(a,b)$
        with $a,b$ the size of a periodic pattern of $\tree{a}, \tree{b}$ respectively.   
        For each $\tree{a}_i \in \Shift{\tree{a}}$ we deduce that
        $\shift{\tree{x}}{1} \shift{\tree{a}_i}{1} = \shift{\tree{x}\tree{a}_i}{1}$,
        which belongs to $\unroll{B}$. 
        And, by union, $\shift{\tree{x}}{1} \Shift{\tree{a}} \subseteq \unroll{B}$.
        By induction, it follows that $\shift{\tree{x}}{k} \Shift{\tree{a}} \subseteq \unroll{B}$, for each $k$. 
        Thereby, $\Shift{\tree{x}} \tree{a} \subseteq \unroll{B}$. 
        Indeed, each element of $\Shift{\tree{x}} \tree{a}$ has multiplicity $1$
        because if there exist two elements $\tree{t}$ and $\tree{u}$ in $\Shift{\tree{x}}$
        such that $\tree{t} \tree{a} = \tree{u} \tree{a}$, by \cite[Lemma 24]{gadouleau2022factorisation_dds},
        we have that $\tree{t} = \tree{u}$.
    \end{proof}
    
    To continue our reasoning, given a forest $\forest{x}$ we denote $\min(\forest{x})$
    the minimum tree among $\forest{x}$, according to the total order introduced in~\cite{gadouleau2022factorisation_dds}.
    Recall that this order is compatible with the product,
    therefore $\min(\forest{x})$ will allow us to identify some factor of a product.
    Lemma~\ref{lemma:quotient_unroll1} will be the induction step for Lemma~\ref{lemma:quotient_unroll2},
    which will consist in: handle the $\min$, remove it, handle the next $\min$, remove it...

    \begin{lemma}\label{lemma:quotient_unroll1}
        Let $A,B$ be two FDDS. 
        If there exists $\forest{x}$ such that $\unroll{A} \forest{x} = \unroll{B}$ then $\min(\forest{X})$ is an unroll tree and $\Shift{\min(\forest{X})} \subseteq \forest{x}$.
    \end{lemma}
    
    \begin{proof}
        We assume that there exists $\forest{X}$ such that $\unroll{A} \forest{x} = \unroll{B}$. 
        We set $\tree{x} = \min(\forest{X})$, $\tree{a}=\min(\unroll{A})$ and $\tree{b}=\min(\unroll{B})$.
        Thus, $\unroll{A} \tree{x}_i \subseteq \unroll{B}$ for all $\tree{x}_i \in \forest{x}$. 
        By Lemma~\ref{lemma:shift}, it follows that
        $\tree{a}_j \Shift{\tree{x}_i} \subseteq \unroll{B}$
        for all $\tree{a}_j \in \unroll{A}$ and $\tree{x}_i \in \forest{X}$. 
        In addition, from the minimality of $\tree{b}$, we deduce that $\tree{a} \tree{x} = \tree{b}$. 
        Thus, $\shift{\tree{a}}{i} \shift{\tree{x}}{i} \in \unroll{B}$ for all positive integer $i$. 
        
        Let $i$ be an integer and $\tree{a}', \tree{x}'$ be two trees of $\unroll{A}$ and $\forest{x}$ respectively such that $\tree{a}' \tree{x}' = \shift{\tree{a}}{i} \shift{\tree{x}}{i}$. 
        So, as already explained $\tree{a}_j \Shift{\tree{x}'} \subseteq \unroll{B}$ for all $\tree{a}_j \in \unroll{A}$. 
        It follows that $\shift{\tree{a}'}{-i \mod a'} \shift{\tree{x}'}{-i \mod x'} = \tree{a} \tree{x}$,
        with $a'=|\Shift{\tree{a}'}|$ the period of $\tree{a}'$
        and  $x'=|\Shift{\tree{x}'}|$ the period of $\tree{x}'$. 
        Thereby, by the minimality of $\tree{a}$, we have that $\shift{\tree{a}'}{-i} \ge \tree{a}$. 
        This implies that $\shift{\tree{x}'}{-i} \le \tree{x}$. 
        However, by the minimality of $\tree{x}$, $\shift{\tree{x}'}{-i} \ge \tree{x}$. 
        We deduce that $\shift{\tree{x}'}{-i} = \tree{x}$ and $\shift{\tree{a}'}{-i} = \tree{a}$.
        We conclude that $\Shift{\tree{x}} \subseteq \forest{x}$.
    \end{proof}
    
    We are now able to conclude our first step, proving
    that a forest $\forest{x}$ result of the division is always an unroll.

    \begin{lemma}\label{lemma:quotient_unroll2}
        Let $A,B$ be two FDDS.
        If there exists $\forest{x}$ such that $\unroll{A} \forest{x} = \unroll{B}$ then there exists $Y$ an FDDS such that $\unroll{Y} = \forest{X}$.
    \end{lemma}
    
    \begin{proof}
        We assume that there exists $\forest{x}$ such that $\unroll{A} \forest{x} = \unroll{B}$. 
        Let $\tree{a}, \tree{x}, \tree{b}$ the minimal tree of $\unroll{A}, \forest{X}$ and $\unroll{B}$ respectively. 
        Thus, by Lemma~\ref{lemma:quotient_unroll1}, we have that $\Shift{\tree{x}} \subseteq \forest{X}$. 
        So, $\unroll{A} \Shift{\tree{x}} \subseteq \unroll{B}$ and there exists an FDDS $C$ such that $\unroll{C} = \Shift{\tree{x}}$. 
        Thereby, since $\unroll{A} \unroll{C} = \unroll{AC}$ we deduce that $\unroll{B} - \unroll{A}\unroll{C}$ is always an unroll. 
        Therefore, we can apply the same reasoning on $\unroll{A} (\forest{X} - \Shift{\tree{x}}) = \unroll{B} - \unroll{A} \Shift{\tree{x}}$. 
        Step by step, it follow that $\forest{X}$ is a sum of shifts.
        And since for each shift $\forest{S}$, there exists a connected FDDS $F$ such that $\unroll{F} = \forest{S}$, we conclude that there exists an FDDS $Y$ such that $\unroll{Y} = \forest{X}$.
    \end{proof}

    Theorem~\ref{theorem:cara_unroll_div} explains,
    given the result $\forest{y}$ of a division,
    how to reconstruct $\forest{x}$ such that the cut at depth $n$ of $\forest{x}$ is the quotient.
    From the cuts we construct a forest, and then apply Lemma~\ref{lemma:quotient_unroll2}
    to obtain an unroll.
    There is a square on $\alpha$, to exceed the least common multiple of two periods
    themselves no greater that $\alpha$.
    
    \begin{theorem}\label{theorem:cara_unroll_div}
        Let $A,B$ be two FDDS and $n = 2 \alpha^2 + \depth{\unroll{B}}$ with $\alpha$ the number of trees in $\unroll{B}$. 
        If there exists a forest $\forest{y}$ such that $\cutUn{A} \forest{y} = \cutUn{B}$ then there  exists an FDDS $X$ such that $\cutUn{X} = \forest{y}$ and
        $\unroll{A} \unroll{X} = \unroll{B}$.
    \end{theorem}
    
    Note that the other direction is given by Proposition~\ref{prop:divideForestFini}.
    
    \begin{proof}
    	Let $\forest{y}$ be a forest such that $\cutUn{A} \cutFn{y} = \cutUn{B}$. 
        Without loss of generality, we can assume that $\forest{y}$ is a finite forest of finite trees.
        Let $\tree{y} \in \forest{y}$. 
        Since $n$ is large enough, by Lemma~\ref{lemma:depth_for_dif}, we deduce that we can identify the subset $\forest{B}_y$ of $\unroll{B}$ such that $\cutUn{A} \tree{y} = \cut{\forest{B}_y}{n}$.  
        We prove that there exists an unroll tree $\tree{x}$ such that $\unroll{A} \tree{x} = \forest{B}_y$. 
        
        We consider $\tree{a} \in \unroll{A}$ and $\tree{b} \in \forest{B}_x$ such that $\cutTn{a} \tree{y} = \cutTn{b}$.
        From the hypothesis, we deduce that $\alpha_A \le \alpha$, with $\alpha_A$ the number of unroll tree in $\unroll{A}$. 
        However, by the definition of unroll, each tree of $\unroll{A}$ has a periodic pattern with size at most $\alpha_A$.
        Also, each tree in $\unroll{B}$ has a periodic pattern with size at most $\alpha$. 
        Thus, $\alpha^2 \ge \alpha_A \alpha \ge \lcm(|\Shift{\tree{a}'}|,|\Shift{\tree{b}'}|)$ for each $\tree{a}' \in \unroll{A}$ and $\tree{b}' \in  \unroll{B}$.
        We set $m = \lcm(|\Shift{\tree{a}}|,|\Shift{\tree{b}}|)$.
        
        Furthermore, by a similar reasoning asin Lemma~\ref{lemme:RecoverPeridicsPart},
        it follows that we can find a sequence $S = (y_0,\ldots, y_{\alpha^2})$ of node of $\tree{y}$
        such that $y_i$ has depth $i$, $y_{i+1}$ is a predecessor of $y_i$, and
        $S \subseteq \bigcap_{B \in max} B$ with $max$ the set of branch of $\tree{y}$ with maximal depth. 
        Remark that, since $\forest{y}$ is finite, $max\in\N$.
        Now, by the periodicity, $\shift{\tree{a}}{m} = \tree{a}$ and $\shift{\tree{b}}{m} = \tree{b}$.
        Therefore, $\cut{\shift{\tree{a}}{m}}{n} \tree{y} = \cut{\shift{\tree{b}}{m}}{n}$. 
        This implies that
        $\cut{\shift{\tree{a}}{m}}{n-m} \cut{\tree{y}}{n-m} =\cut{\shift{\tree{b}}{m}}{n-m}$,
        and so the tree rooted in $y_m$ is $\cut{\tree{y}}{n-m}$.
        We conclude that the tree $\tree{y}_1$ built from $\tree{y}$ where we replace the tree rooted
        in $y_{m}$ by a copy of $\tree{y}$ is such that $\cutT{a}{n+m} \tree{y}_1 \cutT{b}{n+m}$. 
        Step by step, we construct a tree $\tree{x}$ such that $\tree{a} \tree{x} = \tree{b}$.
        In addition, from the construction, $\tree{x}$ admits a periodic pattern
        $P = (T_0, \ldots, T_{\lcm(|\Shift{\tree{a}}|,|\Shift{\tree{b}}|) - 1})$
        such that $T_i$ is the tree rooted in $y_i$ where we remove the tree rooted
        in $y_{i+1}$ for each $i$. 
        
        Let $\tree{a}_i$ be a tree of $\unroll{A}$, $\tree{b}_i$ the tree of $\forest{B}_y$
        such that $\cut{\tree{a}_i}{n} \tree{y} = \cut{\tree{b}_i}{n}$, and $\tree{x}_i$
        the unroll tree built from $\tree{a}_i$ and $\tree{b}_i$ as explained previously.
        It follows that each $\tree{x}_i$ has a periodic pattern of size inferior to $\alpha^2$. 
        As a consequence, by Lemma~\ref{lemma:depth_for_dif}, we deduce that
        $\tree{x}_i = \tree{x}_j$ for all $i,j$.
        We conclude that $\unroll{A} \tree{x}_0 = \forest{B}_y$. 
        
        We can apply the same procedure, for all trees in $\forest{y}$ in order to construct
        a forest $\forest{x}$ of unroll trees such that $\unroll{A} \forest{x} = \unroll{B}$. 
        Finally, by Lemma~\ref{lemma:quotient_unroll2}, there exists $C$ an FDDS such that
        $\unroll{C} = \forest{X}$. 
    \end{proof}
    
    The constructive proof of Theorem~\ref{theorem:cara_unroll_div}
    gave a polynomial time algorithm.

    \begin{theorem}
            Let $A,B$ be two FDDS. 
            We can find in polynomial time a FDDS $X$ such that $\unroll{A} \unroll{X} = \unroll{B}$, if any exists.
    \end{theorem}
    
    \begin{proof}
        From the Theorem~\ref{theorem:cara_unroll_div}, given two unrolls $\forest{A}$ and $\forest{B}$ for compute the unroll $\forest{X}$ such that $\forest{AX} = \forest{B}$, we can just find an forest $\forest{y}$ such that, for a certain $n$, $\cutFn{A} \forest{y} = \cutFn{B}$. 
        From this, since given $A,B$ two FDDS whose the size of $A$ is bounded by $m$ the size de $B$, we can construct $\cutUn{A}$ and $\cutUn{B}$ in $\bigo{m^5}$. 
        Indeed, the $n$ is in $\bigo{m^2}$. 
        In addition, since the size of of each tree in  $\cutUn{A}$ and $\cutUn{B}$ is in $\bigo{m^4}$, we deduce that the number of node in $\cutUn{A} + \cutUn{B}$ is in $\bigo{m^5}$ with $m$ the number of node in $B$.
        Therefore, since the division of $\cutUn{B}$ by $\cutUn{A}$ can be realiazed in cubic time, we conclude that the devision of $\cutUn{B}$ by $\cutUn{A}$ are in $\bigo{m^{15}}$. 
        Finally, since we have identified a periodic pattern of each tree in $\forest{y}$, we can just :
        \begin{enumerate}
        	\item chose a tree $\tree{y}$ in $\forest{y}$,
        	\item find $p$ the size of its smallest period pattern,
        	\item compute $Y$ the roll of $\tree{y}$ to period  $p$,
        	\item remove $\cutUn{Y}$ into $\forest{y}$. 
        \end{enumerate} 
        And since the size of each $\tree{y} \in \forest{y}$ are in $\bigo{m^4}$, it follow that compute $\cutUn{Y}$ are in $\bigo{m^5}$ and $\cutUn{Y}$ has size in $\bigo{m^5}$.
        Thus remove $\cutUn{Y}$ into $\forest{y}$ are obviously in $\bigo{m^{10}}$·
        Finally, if the size of $A$ is not bounded by the size of $B$, obviously $\unroll{A}$ does not divide $\unroll{B}$. 
        The theorem follow. 
    \end{proof}
    
    Although polynomial, the complexity of the previous method is huge ($\bigo{n^{15}}$).
    This is a consequence of the quadratic depth employed for the cut. 
    We now introduce another method, considering a cut at linear depth (in $\bigo{m}$),
    requiring an inexpensive additional check.
    We first need a purely arithmetics lemma, very intuitive
    (each period $p$ of a sequence is a multiple of a unique smallest period).

    \begin{lemma}
        If an unroll tree $\tree{x}$ admits two periodic patterns of size $p_1$ and $p_2$,
        then $\tree{x}$ admits a periodic pattern of size $\gcd(p_1,p_2)$.
    \end{lemma}
    
    \begin{proof}
        We assume that $\tree{x}$ admit a periodic pattern of size $p_1$ and $p_2$. 
        So, either $p_1 = p_2$ and the lemma follow or $p_1 \neq p_2$. 
        Whitout loss of generality, we assume that $p_1 < p_2$. 
        Thus, $\shift{\tree{x}}{i} = \shift{\tree{x}}{k p_1 + i \mod p_2}$ for all $i$ between $0$ and $p_1 - 1$ and positive integer $k$. 
        For all $i$ between $0$ and $p_1 - 1$, we set:
        \[
            C_i = \{ k p_1 + i \mod p_2 \mid k \in \N \}.
        \]
        We deduce that:
        \[
          |C_i| = \cfrac{\lcm(p_1,p_2)}{p_1}.
        \] 
        Also, if $C_i \cap C_j \neq \emptyset$ then $C_i = C_j$. 
        From this, we have that $C_0 \cup \ldots \cup C_{\gcd(p_1,p_2) - 1} = \{0, \ldots, p_2 - 1\}$.
        And also, $C_i = C_{(i + \gcd(p_1,p_2)) \mod p_1}$ for all $i$ between $0$ and $p_1 - 1$.
        Besides, by the construction of the $C_i$, it follow that $\shift{\tree{x}}{j} = \shift{\tree{x}}{j'}$ for all $j,j' \in C_i$. 
        We conclude that $\tree{x}$ has a periodic pattern of size $\gcd(p_1,p_2)$.
    \end{proof}

    Recall that $|\Shift{\tree{x}}|$ is the smallest period of $\tree{x}$.
    
    \begin{corollary}\label{cor:min_per_div_all_per}
        $|\Shift{\tree{x}}|$ is a divisor of each size of periodic pattern of $\tree{x}$.
    \end{corollary}
    
    Because divisors are not greater than the value they divide (for integers),
    we obtain a bound on the period of the trees $\tree{x}\in\forest{X}$
    and $\tree{a}\in\unroll{A}$, in terms of the product of the minimal trees of their shifts.

    \begin{lemma}\label{lemma:period_bound}
        Let $A,B$ be two FDDS. 
        If there exists $\forest{x}$ such that $\unroll{A}\forest{X} = \unroll{B}$,
        then for all $\tree{a} \in \unroll{A}$ and $\tree{x} \in \forest{X}$
        the integers $|\Shift{\tree{a}}|$ and $|\Shift{\tree{x}}|$
        are divisors of $|\Shift{\min(\Shift{\tree{a}})\min(\Shift{\tree{x}})}|$.
    \end{lemma}
    
    \begin{proof}
        We assume that there exists $\forest{x}$ such that $\unroll{A} \forest{X} = \unroll{B}$. 
        Then, by Lemma~\ref{lemma:quotient_unroll2}, there exists a FDDS $Y$ such that $\unroll{Y} = \forest{X}$. 
        Thus, $\Shift{\tree{a}} \Shift{\tree{x}} \subseteq \unroll{B}$ for each $\tree{a} \in \unroll{A}$ and $\tree{x} \in \forest{X}$. 
        Let $\tree{a} \in \unroll{A}, \tree{x} \in \forest{X}$.
        We set $i$ (resp. $j$) the positive integer such that $\shift{\tree{a}}{i} = \min(\Shift{\tree{a}})$ (resp. $\shift{\tree{x}}{j} = \min(\Shift{\tree{x}})$), $\tree{b} \in \unroll{B}$ the tree such that $\shift{\tree{a}}{i} \shift{\tree{x}}{j} = \tree{b}$ and $p_b = |\Shift{\tree{b}}|$. 
        It follow that $\shift{\tree{a}}{i + p_b}\shift{\tree{x}}{j + p_b} = \shift{\tree{b}}{ p_b} = \tree{b}$.
        From the minimality of $\shift{\tree{a}}{i}$, we have that $\shift{\tree{a}}{i + p_b} \ge \shift{\tree{a}}{i}$. 
        Thereby, two cases are possibles. 
        First, $\shift{\tree{a}}{i + p_b} = \shift{\tree{a}}{i}$. 
        So, since $\shift{\tree{a}}{i + p_b}\shift{\tree{x}}{j + p_b} =  \shift{\tree{a}}{i} \shift{\tree{x}}{j}$, by cancellability~\cite[Lemma 24]{gadouleau2022factorisation_dds}, it follow that $\shift{\tree{x}}{j + p_b} = \shift{\tree{x}}{j}$. 
        So, $\tree{a}$ and $\tree{x}$ have a periodic pattern of size $p_b$. 
        Thus, by the Corollary~\ref{cor:min_per_div_all_per}, we conclude that $|\Shift{\tree{a}}|$ and $|\Shift{\tree{x}}|$ are divisors of $p_b$.
        Finally, $\shift{\tree{a}}{i + p_b} > \shift{\tree{a}}{i}$. 
        So, since $\shift{\tree{a}}{i + p_b}\shift{\tree{x}}{j + p_b} =  \shift{\tree{a}}{i} \shift{\tree{x}}{j}$, we have that $\shift{\tree{x}}{j + p_b} < \shift{\tree{x}}{j}$. 
        Contradiction with the minimality of $\shift{\tree{x}}{j}$.
    \end{proof}
    
    \begin{algorithm}
\caption{\texttt{unroll division}}\label{algo:unroll_div}
\begin{algorithmic}[1]
\Require $A,B$ two FDDS
\Ensure A FDDS $X$ such that $\unroll{A} \unroll{X} = \unroll{B}$ if any exists, otherwise $\perp$.
\State $\alpha \gets $ numbers of periodic nodes in $B$
\State $n \gets 2 \times \alpha + \depth{B}$
\State Compute $\cutUn{A}$, $\cutUn{B}$ and $M_A := \{\cut{\min(\Shift{\tree{a}})}{n} \mid  \tree{a} \in \unroll{A}\}$.
\State Compute $P$ the table, indexed by trees of $\cutUn{A}$ and $\cutUn{B}$ such that $P[\tree{x}]$ is the smallest size of period pattern of $\tree{x}$.
\State $\forest{X} \gets \texttt{divide}(\cutUn{B},\cutUn{A})$\label{line:callDivide} 
\If{$\forest{X}$ does not exist}
    \Return $\perp$
\EndIf
\State Sort $\forest{X}$.
\State $Sol \gets \emptyset$ 
\While{$\forest{x} \neq \emptyset$}
	\State $\tree{x} \gets \forest{x}[0]$
    \ForAll{$\tree{a} \in M_A$}
        \State $\tree{b} \gets \tree{ax}$
        \If{$P[\tree{b}] \mod P[\tree{a}] \neq 0 $}
            \State \Return $\perp$
        \EndIf
    \EndFor
    \State $b \gets P[M_A[0] \tree{x}]$
    \State $PP \gets \emptyset$
    \State traverse $\tree{x}$ by associating each vertex with the depth of the tree which is rooted to it
    \ForAll{$i$ between $0$ and $b-1$}
        \State $\tree{x}' \gets $ the tree of $\tree{x}$ with maximal depth
        \State $PP \gets PP \cup \rf{\dt{\tree{x}} - \tree{x}'}$
        \State $\tree{x} \gets \tree{x}'$
    \EndFor
    \State Compute the roll $X$ of $PP$ at its smallest period
    \State $\forest{Y} \gets \cutUn{X}$
    \State Sort $\forest{Y}$
    \If{$\forest{Y} \nsubseteq \forest{X}$}
    	\State \Return $\perp$
    \EndIf
    \State $Sol \gets Sol \cup X$
    \State $\forest{X} \gets \forest{X} - \forest{Y}$
\EndWhile
\State \Return $Sol$
\end{algorithmic}
\label{algo:unroll_div2}
\end{algorithm}

    Algorithm~\ref{algo:unroll_div} implements the condition given by Lemma~\ref{lemma:period_bound},
    which costs a simple extra comparison (this is the inexpensive condition mentioned above),
    but reduces the complexity from $\bigo{m^{15}}$ to $\bigo{m^{9}}$.
    
    \begin{theorem}\label{th:algo_unroll_div_correct_an_poly}
        Algorithm~\ref{algo:unroll_div} is correct for the division of unrolls,
        and runs in $\bigo{m^9}$ time with $m$ the size of its inputs (two FDDS encoding the unrolls).
    \end{theorem}
    
    \begin{proof}
        Let $A,B$ be two FDDS.
        We assume that the algorithm return $Sol$. 
        Therefore $\cutUn{A}\cutUn{Sol} = \cutUn{B}$ and since $n$ is large enough, by Proposition~\ref{prop:divideForestFini}, we conclude that $\unroll{A}\unroll{Sol} = \unroll{B}$. 
        
        Now, we assume that there exists a FDDS $Y$ such that $\unroll{A} \unroll{Y} = \unroll{B}$.
        Then, $\cutUn{A} \cutUn{Y} = \cutUn{B}$.
        In addition, by Lemma~\ref{lemma:period_bound}, we deduce that $|\Shift{\tree{y}}|$ is a divisor of $|\Shift{\tree{b}}|$ for each $\tree{y} \in \unroll{Y}$ and $\tree{b} \in \{\min(\Shift{\tree{a}}) \mid \tree{a} \in \unroll{A}\} \times \min(\Shift{\tree{y}})$. 
        This implies that $\tree{y}$ admits a periodic pattern of size $|\Shift{\tree{b}}|$. 
        We conclude that $Sol$ is correct.
        
        For the complexity, since $n$ is in $\bigo{m}$, we deduce that the size
        of $\cutUn{A}$ and $\cutUn{B}$ are in $\bigo{m^3}$ and line~3 can be executed in $\bigo{m^3}$ time. 
        Thus line~5, 
        the division of $\cutUn{B}$ by $\cutUn{A}$
        using the algorithm from~\cite{gadouleau2022factorisation_dds}, requires $\bigo{m^9}$ time.
        Moreover, since the size of $\forest{X}$ is bounded by the size of $\cutUn{B}$, it follow that line~8 is $\bigo{m^9}$ time.
        Let $\forest{X} = \tree{x}_0 + \ldots + \tree{x}_k$ and $\{\tree{b}_0, \ldots, \tree{b}_k\}$ be a subset of $\cutUn{B}$ such that $\tree{a} \tree{x}_i = \tree{b}_i$ for all $i$ where $\tree{a}$ is a tree of $\{\min(\Shift{\tree{t}}) \mid \tree{t} \in \unroll{A}\}$. 
        Therefore, the number of nodes in $\tree{x}_i$ is less than or equal to the number of nodes in $\tree{b}_i$. 
        The complexity of the $i$-th iteration of the loop in lines~10--32
        is due to lines~27--33, namely $\bigo{m^4 \log m}$.
        The complexity of lines~20--26 is $\bigo{m_i}$,
        with $m_i$ the number of nodes in $\tree{b}_i$. 
        Indeed, the operations in these lines require only two breadth first searches of $\tree{x}_i$
        (the global pre-processing of line~20 and post-processing of line~26 improves the complexity of iterating line~22 to a linear total time;
        at line~23 a tree is identified with its root, hence for each operation we need a simple traversal of the neighbors).  
        The complexity of lines~11--19 is $\bigo{\alpha_A}$ with $\alpha_A$ the number of trees in $\cutUn{A}$.
        Finally, we need to add the complexity of lines~27--33 themselves.
        Since $\tree{x}_i$ is a tree of $\cutUn{X}$, we deduce that line~27 can be performed in $\bigo{m_i \sqrt{m_i}}$ time and the size of $\cutUn{X}$ is also in $\bigo{m_i\sqrt{m_i}}$. 
        This is due the face that, if $\cutTn{t}{n}$ has size in $\bigo{s}$, then the number of trees in $\cut{\Shift{\tree{t}}{n}}$ is in $\bigo{\sqrt{s}}$.  
        Thus, the complexity of line~28 is in $\bigo{m_i\sqrt{m_i} \log m_i}$ and the lines~29 and 33 are $\bigo{m^3}$ time.
        
        In summary, the complexity of the $i$-th iteration of the loop in lines~10--27 is $\bigo{\alpha_A + m_i + m_i\sqrt{m_i} \log m_i + m^3}$.
        Thus the total complexity of the loop is $\bigo{\sum_{i=0}^{k} \alpha_A + m_i + m_i\sqrt{m_i} \log m_i + m^3}$.
        However, $\sum_{i=0}^{k} \alpha_A = \alpha \le m$ and $\sum_{i=0}^{k} m_i \sqrt{m_i} \le m^3$.
        We deduce that $\bigo{\sum_{i=0}^{k} \alpha_A + m_i m_i\sqrt{m_i}\log m_i + m^3} = \bigo{m^4 \log m}$. 
        We conclude that complexity of the algorithm is $\bigo{m^9}$.
    \end{proof}

    To conclude this section, let us present two examples
    where our method works beyond the scope considered so far,
    that is, when our hypothesis on the connectedness of
    $\forest{x}$ (the result of the division) is weakened.
    
    We define that a FDDS $A$ is \emph{component-minimal} when
    the number of different shifts in $\unroll{A}$
    (\emph{i.e.}, the size of
    $\{\Shift{\tree{a}} \mid \tree{a}\in\unroll{A}\}$)
    is equal to the number of connected components of $A$. 
    Expressing the shifts of one component from its minimal unroll tree,
    a direct consequence of this definition is that, if we denote
    $A = A_0 + \ldots + A_n$ and
    $\unroll{A} = k_0 \Shift{\min(\unroll{A_0})} + \ldots + k_n \Shift{\min(\unroll{A_n})}$,
    then $\unroll{A_i} = k_i \Shift{\min(\unroll{A_i})}$ for each $i$.
    In other words, we can identify which unroll tree is rolled into which connected component of $A$.
    
    Symmetrically, we define that a FDDS $A$ is \emph{component-maximal} when
    the number of minimal tree copy of a shift in $\unroll{A}$
    (\emph{i.e.}, the size of the multiset $\{\{\tree{a}\mid \tree{a}\in\unroll{A}\text{ and }\tree{a}=\min(\Shift{\tree{a}})\}\}$)
    is equal to the number of connected components in $A$.
    In this case, for each copy of a shift, there is a connected component in $A$.
    A direct consequence of this definition is that, if we denote
    $\unroll{A} = k_0 \Shift{\tree{t}_0} + \ldots + k_n \Shift{\tree{t}_n}$,
    then we have $A = k_0 A_0 + \ldots + k_n A_n$ with $\unroll{A_i} = \Shift{\tree{t}_i}$ for each $i$.
    In other words, when we roll an unroll tree and obtain twice the same dynamics,
    then they belong to two different connected components of $A$.
    
    From Theorem~\ref{th:algo_unroll_div_correct_an_poly},
    we deduce that we can solve in polynomial time the equations of the form $A X = B$,
    when $X$ is component-minimal or component-maximal. 
    Indeed, for these two cases the rolls of $\unroll{X}$ are injective,
    and the following algorithm computes $X$.
    \begin{algo}Given two FFDS $A$ and $B$, we can compute $X$ such that $X$ is component-minimal
      or component-maximal and $AX=B$ (if any exists) by:
	    \begin{enumerate}
              \item computing $\unroll{X} = \cfrac{\unroll{A}}{\unroll{B}}$ with Algorithm~\ref{algo:unroll_div2};
	      \item finding the smallest periodic pattern of each tree in $\unroll{X}$;
	      \item rolling $\unroll{X}$ with the chosen period
	        (for component-minimal it is twice the number of copy of the shift,
	        and for component-maximal it is simply $2$);
	      \item and checking if the $A$ multiplied by the resulting FDDS is equal to $B$.
	    \end{enumerate}
    \end{algo}
    
    From the preceding exposure, the (difficult) remaining case is when, in $X$,
    a given shift appears $k$ times,
    but the number of connected components having this shift is strictly between $1$ and $k$.
    In this case, we do not know how to efficiently choose
    the period of roll for each unroll tree.

	\section{Conclusions}

In this article we have proven the cancellation property for products of unrolls
and established that the division of FDDSs is polynomial-time when searching for connected quotients only. 
Furthermore, we have proven that calculating the $k$-th root of a FDDSs is polynomial-time
if the solution is connected.
Moreover, we have shown that solving equations of the form $A X^k = B$ is polynomial if $X$ is connected.
Finally, we have demonstrated that we can compute in polynomial time a solution
to the equation $\unroll{A} \forest{x} = \unroll{B}$, and we have also presented a polynomial method
to compute the solution of $AX = B$ for two special form of $X$ (component-minamal and component-maximal). 
However, numerous questions remain unanswered.

The main direction for further investigation involves removing the connectivity condition as we have made at the end of Section~\ref{section:divison_unroll}.
The cancellation property of unrolls we proved, and the new polynomial-time algorithm for the division, suggest that the primary challenge for FDDS division lies in searching for the correct period of roll
rather than to consider the transients.
Another intriguing direction is solving general polynomial equations $P(X_1, \ldots, X_n) = B$
with a constant right-hand side $B$. While this appears to be at least as challenging as division,
some specific cases could yield more direct results.
For example, it was proved that $P(X) = B$ is polynomial time if $P$ is an injective polynomial \cite{poly_inj}.
Furthermore, the results of this work can improve the state of the art of the solution of $P(X_1, \ldots, X_n) = B$ where the polynomial $P$ is a sum of univariate monomials.
Indeed, a technique to solve (and enumerate the solutions) of this type of equations
in a finite number of systems of equations of the form $AX^k = B$ has been introduced.
Thus, our result, which is more efficient than previously known techniques,
can have a positive impact on the complexity of the proposed pipeline.
It would also be interesting to investigate whether our techniques also apply to finding nontrivial
solutions to equations of the form $XY = B$ with $X$ and $Y$ connected,
which would make it possible to increase our knowledge on the problem of irreducibility.

	\begin{credits}
        \subsubsection{\ackname}
          SR was supported by the French Agence Nationale pour la Recherche (ANR)
          in the scope of the project ``REBON'' (grant number ANR-23-CE45-0008)
          KP, AEP and MR by the EU project MSCA-SE-101131549 ``ACANCOS'',
          and KP, AEP, SR and MR by the project ANR-24-CE48-7504 ``ALARICE''.
        \end{credits}
    
	\bibliography{biblio}

\begin{thebibliography}{10}
\providecommand{\url}[1]{\texttt{#1}}
\providecommand{\urlprefix}{URL }
\providecommand{\doi}[1]{https://doi.org/#1}

\bibitem{bernot2013computational_biology}
Bernot, G., Comet, J.P., Richard, A., Chaves, M., Gouz\'e, J.L., Dayan, F.:
  Modeling in computational biology and biomedicine. In: Modeling and Analysis
  of Gene Regulatory Networks: A Multidisciplinary Endeavor, pp. 47--80.
  Springer (2012)

\bibitem{dorigatti2018polynomial}
Dennunzio, A., Dorigatti, V., Formenti, E., Manzoni, L., Porreca, A.E.:
  Polynomial equations over finite, discrete-time dynamical systems. In:
  Cellular Automata, 13th International Conference on Cellular Automata for
  Research and Industry, ACRI 2018. Lecture Notes in Computer Science, vol.
  11115, pp. 298--306. Springer (2018)

\bibitem{dennunzio2019solving}
Dennunzio, A., Formenti, E., Margara, L., Montmirail, V., Riva, S.: Solving
  equations on discrete dynamical systems. In: Computational Intelligence
  Methods for Bioinformatics and Biostatistics, 16th International Meeting,
  {CIBB} 2019. Lecture Notes in Computer Science, vol. 12313, pp. 119--132.
  Springer (2019)

\bibitem{dennunzio2023dds_journal}
Dennunzio, A., Formenti, E., Margara, L., Riva, S.: An algorithmic pipeline for
  solving equations over discrete dynamical systems modelling hypothesis on
  real phenomena. Journal of Computational Science  \textbf{66},  101932 (2023)

\bibitem{dennunzio2024resolution_permutation}
Dennunzio, A., Formenti, E., Margara, L., Riva, S.: A note on solving basic
  equations over the semiring of functional digraphs. arXiv e-prints  (2024),
  \url{https://arxiv.org/abs/2402.16923}

\bibitem{reaction_systems}
Ehrenfeucht, A., Rozenberg, G.: Reaction systems. Fundamenta Informaticae
  \textbf{75},  263--280 (2007)

\bibitem{formenti2021mdds}
Formenti, E., R{\'e}gin, J.C., Riva, S.: {MDD}s boost equation solving on
  discrete dynamical systems. In: International Conference on Integration of
  Constraint Programming, Artificial Intelligence, and Operations Research. pp.
  196--213. Springer (2021)

\bibitem{gadouleau2011graph_entropy}
Gadouleau, M., Riis, S.: Graph-theoretical constructions for graph entropy and
  network coding based communications. IEEE Transactions on Information Theory
  \textbf{57}(10),  6703–6717 (2011)

\bibitem{gershenson2004random_bn}
Gershenson, C.: Introduction to random boolean networks. arXiv e-prints
  (2004), \url{https://doi.org/10.48550/arXiv.nlin/0408006}

\bibitem{automata_book}
Goles, E., Mart\`inez, S.: Neural and Automata Networks: Dynamical Behavior and
  Applications. Kluwer Academic Publishers (1990)

\bibitem{hammack2011graph_product_book}
Hammack, R., Imrich, W., Klav\v{z}ar, S.: Handbook of Product Graphs. Discrete
  Mathematics and Its Applications, CRC Press, second edn. (2011)

\bibitem{hopcroft1974linear}
Hopcroft, J.E., Wong, J.K.: Linear time algorithm for isomorphism of planar
  graphs (preliminary report). In: Proceedings of the sixth annual ACM
  symposium on Theory of computing. pp. 172--184 (1974)

\bibitem{gadouleau2022factorisation_dds}
Naquin, E., Gadouleau, M.: Factorisation in the semiring of finite dynamical
  systems. Theoretical Computer Science  \textbf{998},  114509 (2024)

\bibitem{poly_inj}
Porreca, A.E., Rolland, M.: Injectivity of polynomials functions over finite
  discrete dynamical system. arXiv e-prints  (2025)

\bibitem{strozecki2021enum_complexity}
Strozecki, Y.: Enumeration {C}omplexity: {I}ncremental {T}ime, {D}elay and
  {S}pace (2021), {HDR} thesis, Universit\'e de Versailles
  Saint-Quentin-en-Yvelines

\bibitem{thomas1973genetic_control_circuits}
Thomas, R.: Boolean formalization of genetic control circuits. Journal of
  Theoretical Biology  \textbf{42}(3),  563--585 (1973)

\bibitem{thomas1990biological_feedback}
Thomas, R., D'Ari, R.: Biological Feedback. CRC Press (1990)

\end{thebibliography}
	\bibliographystyle{splncs04}
	
\end{document}